\documentclass[
a4paper,reqno]{amsart}
\usepackage[T1]{fontenc}
\usepackage[english]{babel}
\usepackage{amssymb,bbm,bm,enumerate}
\usepackage{color}
\usepackage[linktocpage=true,colorlinks=true, linkcolor=blue, citecolor=red, urlcolor=green]{hyperref}
\usepackage[all]{xy}

\usepackage{float}
\restylefloat{table}

\usepackage{tikz-cd}

\newcommand{\mc}[1]{\mathcal{#1}}
\newcommand{\mf}[1]{\mathfrak{#1}}
\newcommand{\mb}[1]{\mathbb{#1}}

\newcommand{\ul}[1]{\underline {#1}}
\newcommand{\id}{\mathbbm{1}}

\newcommand{\tint}{{\textstyle\int}}

\DeclareMathOperator{\Mat}{Mat}

\DeclareMathOperator{\res}{Res}

\DeclareMathOperator{\im}{Im}

\DeclareMathOperator{\Span}{Span}

\DeclareMathOperator{\Res}{Res}

\DeclareMathOperator{\ord}{ord}

\theoremstyle{plain}
\newtheorem{theorem}{Theorem}[section]
\newtheorem{lemma}[theorem]{Lemma}
\newtheorem{proposition}[theorem]{Proposition}
\newtheorem{corollary}[theorem]{Corollary}

\theoremstyle{definition}
\newtheorem{definition}[theorem]{Definition}
\newtheorem{example}[theorem]{Example}

\theoremstyle{remark}
\newtheorem{remark}[theorem]{Remark}

\numberwithin{equation}{section}

\definecolor{light}{gray}{.9}

\begin{document}

\title{On Lax operators}

\author{Alberto De Sole}
\address{Dipartimento di Matematica, Sapienza Universit\`a di Roma,
P.le Aldo Moro 2, 00185 Rome, Italy}
\email{desole@mat.uniroma1.it}
\urladdr{www1.mat.uniroma1.it/\$$\sim$\$desole}
\author{Victor G. Kac}
\address{Dept of Mathematics, MIT,
77 Massachusetts Avenue, Cambridge, MA 02139, USA}
\email{kac@math.mit.edu}
\author{Daniele Valeri}
\address{School of Mathematics \& Statistics, University of Glasgow, G12 8QQ Glasgow, UK}
\email{daniele.valeri@glasgow.ac.uk}



\begin{abstract}
We define a Lax operator as a monic pseudodifferential
operator $L(\partial)$ of order $N\geq 1$, such that the Lax
equations $\frac{\partial L(\partial)}{\partial t_k}=[(L^{\frac kN}(\partial))_+,L(\partial)]$
are consistent and non-zero for infinitely many positive integers $k$.
Consistency of an equation means that its flow is defined by
an evolutionary vector field.
In the present paper we demonstrate that the traditional
theory of the KP and the $N$-th KdV hierarchies holds for arbitrary
scalar Lax operators.
\end{abstract}


\maketitle

\tableofcontents

\section{Introduction}
\label{sec:1}

In his seminal paper \cite{Lax68}
Lax observed that the famous KdV equation
$$
u_t
=
\frac14 u_{xxx}+\frac32 uu_x
$$
is equivalent to an equation of the form $\frac{\partial L(\partial)}{\partial t}=[B(\partial),L(\partial)]$, on the differential operator $L(\partial)=\partial^2+u$
for a certain differential operator $B(\partial)$ of order $3$.
This has lead to the Lax equation approach in the theory of integrable systems.
As is well known,
writing an evolution equation in a Lax form allows one to construct its higher symmetries and integrals of motion.

The first beautiful application of this approach was developed by Gelfand and Dickey \cite{GD76}
by considering a differential operator $L(\partial)$ of the form 
\begin{equation}\label{eq:intro1}
L(\partial)=\partial^N+\sum_{j=1}^{N-1}u_j\partial^{N-j-1}
\,\,,\,\,\,\, N\geq2
\,,
\end{equation}
where $u_1,\dots,u_{N-1}$ are the generators of the algebra of differential polynomials $\mc V_{N-1}$ in
$N-1$ differential variables, cf. \eqref{20170721:eq3}.
They showed that the hierarchy of \emph{Lax equations}
\begin{equation}\label{eq:intro2}
\frac{\partial L(\partial)}{\partial t_k}=[B_k(\partial),L(\partial)]
\,,\,\text{ where }\,\,
B_k(\partial)=(L^{\frac kN}(\partial))_+
\,,\,\, k\in\mb Z_{\geq1}
\,,
\end{equation}
is a hierarchy of compatible evolution equations, i.e. equations of the form
$$
\frac{\partial u_j}{\partial t_k}
=
R_{j,k}
\,\,,\,\,\,\,
j=1,\dots,N-1
\,,\,\,
k\in\mb Z_{\geq1}
\,,\,\,
R_{j,k}\in\mc V_{N-1}
\,,
$$
which is called the $N$-th KdV hierarchy.
Recall that the subscript $+$ stands for the differential part of a pseudodifferential operator
and compatibility means that the partial derivatives $\frac{\partial}{\partial t_n}$ commute.
The case $N=2$ of \eqref{eq:intro1}
is the KdV hierarchy, of which \eqref{eq:intro1} for $k=3$ is the KdV equation.
(For arbitrary $N$ and $k=1$ equation \eqref{eq:intro1}
is the trivial evolution equation $\frac{\partial u_j}{\partial t_1}=u_j'$,
and for $k\in N\mb Z$ equation \eqref{eq:intro1} is the zero equation $\frac{\partial u}{\partial t_j}=0$.)

Note however that a Lax equation is not necessarily an evolution equation.
For example, taking $L(\partial)=\partial^3+u$,
we have $B_2(\partial)=(L^{\frac23}(\partial))_+=\partial^2$,
hence the corresponding Lax equation is
$$
\frac{\partial u}{\partial t_2}
=
[B_2(\partial),L(\partial)]
=
2u'\partial+u''
\,,
$$
which is a linear system
$$
u'=0
\,\,,\,\,\,\,
\frac{\partial u}{\partial t_2}
=
u''
\,.
$$
This is not an evolution equation and there exists no evolution equation on $\mc V_1$
whose flow defines the flow of this system.
This equation is thus called ``inconsistent''.

The general definition of ``consistency'' is as follows. Let $u=(u_\alpha)_{\alpha\in I}$ be the column vector 
of differential variables of $\mc V=\mc V_\ell$ and let
$P=(P_\alpha)_{\alpha\in I}\in\mc V^\ell$. 
The \emph{evolution equation} associated to $P$ is defined as
\begin{equation}\label{eq:evolution}
\frac{\partial u}{\partial t}=P\,.
\end{equation}
By the chain rule, this equation induces the evolution of an arbitrary $f\in\mc V$: $\frac{\partial f}{\partial t}=X_P(f)$, where $X_P$ 
is the following derivation of $\mc V$ commuting with $\partial$, called an \emph{evolutionary vector field}:
$$
X_P=\sum_{\alpha\in I}\sum_{n\in\mb Z_{\geq0}}(\partial^nP_\alpha)\frac{\partial}{\partial u_{\alpha}^{(n)}}
\,.
$$
Let now $M(\partial)$ be an $m\times\ell$ matrix differential operator over $\mc V$, and let $Q=(Q_i)_{i=1}^m\in\mc V^m$. 
The corresponding \emph{linear system of quasi-evolution equations} is defined as:
\begin{equation}\label{eq:linear-intro}
M(\partial)\frac{\partial u}{\partial t}=Q
\,.
\end{equation}
We call such system \emph{consistent} if $Q$ lies in the image of $M(\partial):\mc V^\ell\to\mc V^m$, i.e. there exists $P\in\mc V^\ell$ such that
\begin{equation}\label{eq:consistent}
M(\partial)P=Q\,.
\end{equation}
In this case the flow, defined by the evolution equation \eqref{eq:evolution}, is ``consistent'' with the linear system \eqref{eq:linear-intro}.

A consistent linear system \eqref{eq:linear-intro} is called \emph{compatible} with another consistent linear system of quasi-evolution equations
$\tilde M(\partial)\frac{\partial u}{\partial \tilde t}=\tilde Q$, such that $\tilde M(\partial)\tilde P=\tilde Q$,
if $P$ and $\tilde P$ can be chosen in such a way that the corresponding evolutionary vector fields $X_P$ and $X_{\tilde P}$ commute:
\begin{equation}\label{eq:intro3}
X_PX_{\tilde P}=X_{\tilde P} X_P\,.
\end{equation}

In the present paper we consider an arbitrary scalar monic pseudodifferential operator of order $N\geq1$
over the algebra of differential polynomials $\mc V_\ell$ in $\ell$ differential variables ($\ell$ may be infinite)
\begin{equation}\label{eq:intro4}
L(\partial)
=
\partial^N+\sum_{i=0}^{\infty}a_i\partial^{N-1-i}
\,\,,\,\,\,\,
a_i\in\mc V_\ell
\,.
\end{equation}
Given a differential operator $B(\partial)\in\mc V_\ell[\partial]$, the \emph{Lax equation} associated to $L(\partial)$ and $B(\partial)$ is defined as
\begin{equation}\label{eq:intro5}
\frac{\partial L(\partial)}{\partial t_B}
=
[B(\partial),L(\partial)]
\,.
\end{equation}
The Lax equation \eqref{eq:intro5} is a linear system of quasi-evolution equations,
which can be written down explicitly as follows. 
Note that, for a pseudodifferential operator $L(\partial)$ as in \eqref{eq:intro4} and a differential operator
$B(\partial)$ of order $k$, one has
(see Section \ref{sec:3} for details):
$$
[B(\partial),L(\partial)]=\sum_{i=-k}^\infty q_{i}\partial^{N-1-i}\,,\quad q_{i}\in\mc V_{\ell}\,.
$$
Let
$$
D_{i,\alpha}(\partial)=\sum_{n=0}^\infty\frac{\partial a_i}{\partial u_{\alpha}^{(n)}}\partial^n\in\mc V_\ell[\partial]\,,
\quad i\geq0,\,\alpha\in I\,,
$$ 
be the Frechet derivative of the coefficient $a_i\in\mc V_\ell$ of $\partial^{N-1-i}$ in $L(\partial)$. Then, the Lax equation
\eqref{eq:intro5} becomes
the following linear system of quasi-evolution equations
\begin{align}
& 0=q_i\quad\text{for }i\in\{-k,\dots,-1\}
\,,\label{eq:Lax2a} \\
& \sum_{\alpha\in I}D_{i,\alpha}(\partial)\frac{\partial u_\alpha}{\partial t_B}=q_i\quad\text{for }i\geq0\,.
\label{eq:Lax2b}
\end{align}
Hence, this linear system of quasi-evolution equations is consistent if and only if $q_i=0$ for $i<0$ and $Q=(q_i)_{i=0}^{\infty}\in\mc V_\ell^\infty$ lies
in the image of the matrix differential operator $D^{(a)}(\partial)=(D_{i,\alpha}(\partial)):\mc V_\ell^I\to\mc V_\ell^\infty$.

Let $B_k(\partial)=(L^{\frac{k}{N}}(\partial))_+$.
In Section \ref{sec:3} we shall see that condition \eqref{eq:Lax2a} holds if and only if
$B(\partial)$ is, up to an adding an element of $\mc V$, a linear combinations with constant coefficients of the differential
operators $B_k(\partial)$, $k\in\mb Z_{\geq1}$.

Let $\tilde{\mc Z}_L$ (respectively, $\mc Z_L$) denote the set of all positive integers $k$
for which the Lax equation $\frac{\partial L(\partial)}{\partial t_k}=[B_k(\partial),L(\partial)]$
is non-zero (i.e. RHS $\neq0$) and consistent (respectively, is consistent).
Clearly, $\tilde{\mc Z}_L\subset\mc Z_L$. We  say that $L(\partial)$ is a \emph{Lax operator} if the set $\tilde{\mc Z}_L$ is infinite.
For example, the operators \eqref{eq:intro1} are Lax operators for every $N\geq2$, while $L(\partial)=\partial^3+u$
is NOT.

If $L(\partial)$ is a Lax operator, then we get the
corresponding infinite hierarchy of consistent Lax equations
\begin{equation}\label{eq:Laxhier}
\frac{\partial L(\partial)}{\partial t_k}=[B_k(\partial),L(\partial)]\,,\quad k\in\mc Z_L\,. 
\end{equation}
One of the basic results of the theory of integrable systems is that the hierarchy \eqref{eq:Laxhier}
is compatible (see Proposition \ref{20170719:prop}). (Note that, if $B(\partial)\in\mc V$, then the Lax equation \eqref{eq:intro5}
is consistent, but it may fail to be compatible with the equations of the hierarchy \eqref{eq:Laxhier}.)

The next important development in the theory of Lax equations is the work of Sato \cite{Sat81} and his disciples 
\cite{DJKM81,DJKM83} on the KP hierarchy and its analogues and reductions.
The KP hierarchy is defined as the following system of Lax equations on the Sato operator
$L(\partial)=\partial+\sum_{i=1}^{\infty}u_i\partial^{-i}$, 
where $u_1,u_2,\dots$ are the differential variables
of $\mc V_\infty$:
\begin{equation}\label{eq:Laxhier2}
\frac{\partial L(\partial)}{\partial t_k}=[B_k(\partial),L(\partial)]\,,\quad k\in\mb Z_{\geq1}\,,
\end{equation}
for which $\mc Z_L=\mb Z_{\geq1}$.

Around the same time Drinfeld and Sokolov \cite{DS85}
associated to each simple Lie algebra $\mf g$ and its principal nilpotent element $f$
a hierarchy of bi-Hamiltonian PDE.
In particular, for each classical $\mf g$ they constructed the corresponding Lax operator $L(\partial)$.
In the case $\mf g=\mf{sl}_N$, $L(\partial)$ is the Gelfand-Dickey operator \eqref{eq:intro1}
so that the corresponding hierarchy is the $n$-th KdV hierarchy.
For $\mf g=\mf{so}_{2n+1}$ (respectively $\mf{sp}_{2n}$),
$L(\partial)\circ\partial $ (resp. $L(\partial)$) is the ``generic'' skewadjoint (resp. selfadjoint)
monic differential operator of order $2n+1$ (resp. $2n$),
while, for $\mf g=\mf{so}_{2n}$, $L(\partial)\circ\partial$ is the sum of the ``generic'' skewadjoint differential operator of order $2n-1$
and the pseudodifferential operator $u\partial^{-1}\circ u$.

The ideas of \cite{GD76} and \cite{DS85}
were further developed in \cite{DSKV15,DSKV16a,DSKV18}.
In particular, in \cite{DSKV16a} (respectively \cite{DSKV18})
to each nilpotent element of the Lie algebra $\mf g=\mf{sl}_N$
(resp. $\mf{so}_N$ and $\mf{sp}_N$)
corresponding to a partition $\ul p$ of $N$,
we constructed an $r\times r$-matrix monic pseudodifferential operator $L_{\ul p}(\partial)$ of order $p$,
where $p$ is the maximal part of $\ul p$ and $r$ is its multiplicity.
It follow from \cite{DSKV16a} (resp. \cite{DSKV18})
that $L(\partial)=L_{\ul p}(\partial)$ is a Lax operator with $\mc Z_L=\mb Z_{\geq1}$
(resp. $\mc Z_L=1+2\mb Z_{\geq0}$),
and that the Lax equations associated to it are Hamiltonian for the Poisson structure of the corresponding $\mc W$-algebra
$\mc W(\mf g,\ul p)$.

The simplest Lax operator $L_{\ul p}(\partial)$ corresponds to the partition $\ul p=N$,
consisting of one part $N$,
which is associated to the principal nilpotent element of $\mf {sl}_N$
(resp. $\mf{so}_N$ for odd $N$ and $\mf{sp}_N$ for even $N$).
In this cases $L_{\ul p}(\partial)$ coincides with \eqref{eq:intro1},
the Lax operator for the $N$-th Gelfand-Dickey hierarchy
(resp. the Lax operator associated to the $\mf{so}_N$ Drinfeld-Sokolov hierarchies composed with $\partial$,
and the Lax operator associated to the $\mf{sp}_N$ Drinfeld-Sokolov hierarchies).
For $\mf g=\mf{so}_N$, $N$ even, $\ul p=(N-1,1)$
and $L_{\ul p}(\partial)$ is a certain pseudodifferential operator \cite{DS85}.

Another well-known example in the case of $\mf g=\mf{sl}_N$ 
corresponds to the partition $\underline p=(p,1^{N-p})$, where $N>p$. In this case,
after a reduction by non-evolving variables,
$L_{\underline p}(\partial)=L_{(p)}(\partial)+\sum_{i=1}^{N-p}p_i\partial^{-1}\circ q_i$,
which is therefore a Lax operator and the corresponding hierarchy \eqref{eq:Laxhier} of Lax equations is a subsystem
of a hierarchy of Hamiltonian equations \cite[Example 14.1]{CDSKVvdL20}. This produces the well-known $(N-p)$-vector $p$-constrained KP hierarchy
\cite{KSS91,Che92,KS92,SS93}. The case $p=2$ produces the $(N-2)$-component Yajima-Oikawa hierarchy
\cite{YO76,Ma81}.

In the present paper we establish a series of results in the generality
of arbitrary scalar Lax operators, similar to the traditional results for the Sato Lax operator. 
Our proofs are often identical to those of the Kyoto school \cite{DJKM83,Shi86} and of the book \cite{Dic03}, 
sometimes they are simpler and more rigorous. We show that the hierarchy of Lax equations,
associated to a Lax operator $L(\partial)$ is compatible (Proposition \ref{20170719:prop}), and is equivalent to
\begin{enumerate}[(i)]
\item Zakharov-Shabat equations \cite{ZS74} (Proposition \ref{20170720:prop}(a));
\item complementary Zakharov-Shabat equations (Proposition \ref{20170720:prop}(b));
\item Sato equations on the dressing operator (Theorem \ref{thm:sato});
\item the linear problem on the wave function (Theorem \ref{thm:linear});
\item the bilinear equation on the wave function (Theorem \ref{thm:bilinear}).
\end{enumerate}
As for the Sato Lax operator, this leads to the construction of the tau-function for $L(\partial)$. 
The tau-function for all Lax operators $L_{\underline p}(\partial)$, for all partitions $\ul p$, 
were constructed in \cite{CDSKVvdL20}
as those for the $r$-component KP hierarchy satisfying a simple constraint.
Of course, the importance of the tau-function $\tau$ for the Lax operator $L(\partial)$ stems from the fact that,
as for KP, all solutions of the corresponding hierarchy \eqref{eq:Laxhier} can be expressed via $\tau$ (see formula \eqref{eq:befana}).

For the Sato Lax operator formula \eqref{eq:befana}
establishes an essentially bijective correspondence between tau-functions and the solutions of the KP-hierarchy.
For arbitrary $\ul p$-reductions of the KP hierarchy the corresponding tau-functions satisfy a simple constraint,
which allowed for their explicit construction in \cite{CDSKVvdL20}.

Unfortunately, it is still an open problem for an arbitrary Lax operator
to find the constraints on the tau-function imposed by the constraints on the coefficients of $L(\partial)$.
This problem has been solved for the Lax operators of the BKP and CKP hierarchies \cite{DJKM81,CW13,KZ20},
but not, for example, for the Lax operator corresponding to the KN equation discovered by Sokolov \cite{Sok84}.

The above discussion can be extended to $r\times r$ matrix pseudodifferential operators $L(\partial)$.
This will be treated in a forthcoming publication.

\smallskip

Throughout the paper the base field $\mb F$ is a field of characteristic zero.

%

\section{Algebraic setup}
\label{sec:2}

\subsection{
Functions on space-time
}
\label{sec:2.2}

Throughout the paper we let $\mc F$ be a given commutative, associative, unital algebra over $\mb F$
endowed with commuting derivations 
$$
\partial
\,\Big(=\frac{\partial}{\partial x}\Big),\,\frac{\partial}{\partial t_k}\,:\,\,\mc F\to\mc F\,,\quad k\in \mc Z\,,
$$
indexed by an index set $\mc Z\subset\mb Z_{\geq1}$.

The elements of $\mc F$ are called \emph{functions on space-time}
(or simply \emph{functions}),
and will be usually denoted as $f=f(x,\bm t)=f(x,t_k,k\in\mc Z)$
(In the usual terminology, the ``space variable'' is $x$, and there are many 
``time variables'' $t_k$, $k\in\mc Z$.)

We assume that the common kernel of all space and time derivatives is the base field $\mb F$:
\begin{equation}\label{20170721:eq1}
\mb F
=
\big\{c\in\mc F\,\big|\,\partial c=0\,\text{ and }\, \frac{\partial c}{\partial t_k}=0\,\text{ for all }\,k\in\mc Z\big\}
\,.
\end{equation}
We also assume that $\mc F$ is endowed with a surjective algebra homomorphism
$\mc F\twoheadrightarrow\mb F$,
restricting to the identity map on $\mb F\subset\mc F$,
which we shall call the \emph{evaluation at} $x=0,\bm t=0$,
and we shall denote as
\begin{equation}\label{20170718:eq8b}
f=f(x,\bm t)
\,\mapsto\,
f(0)
=f|_{x=0,\bm t=0}
\,\in\mb F
\,.
\end{equation}
\begin{definition}\label{def:integrable}
Given elements $g,f_k\in\mc F$, $k\in\mc Z$,
consider a system of equations on the unknown function $\varphi\in\mc F$:
\begin{equation}\label{20170721:eq2}
\partial\varphi=g
\,\,,\,\,\,\,
\frac{\partial\varphi}{\partial t_k}
=
f_k
\,\text{ for all }\,
k\in\mc Z
\,.
\end{equation}
The system \eqref{20170721:eq2} is called \emph{compatible}
if the following conditions hold:
\begin{equation}\label{20170721:eq2b}
\frac{\partial g}{\partial t_k}
=
\partial f_k
\text{ for all } k\in\mc Z
\,,\,\,\text{ and }\,
\frac{\partial f_k}{\partial t_h}
=
\frac{\partial f_h}{\partial t_k}
\,\text{ for all }\,
h,k\in\mc Z
\,.
\end{equation}
The algebra of functions on space-time $\mc F$ is said to be \emph{integrable}
if, for every $c\in\mb F$ and every compatible system of equations \eqref{20170721:eq2},
there exists a unique solution $\varphi\in\mc F$ such that $\varphi(0)=c$.
\end{definition}
\begin{example}\label{20170724:ex}
The algebra $\mb F[x,t_1,t_2,t_3,\dots]$ of polynomials in infinitely many variables,
and the algebra $\mb F[[x,t_1,t_2,t_3,\dots]]$ of formal power series, are both integrable.
\end{example}

\subsection{
Unknown (dependent) variables
}
\label{sec:2.1}

We let $u_1,\dots,u_\ell$ ($\ell$ may be infinite)
be the ``unknown functions'' on space-time (=dependent variables),
and we let $\mc V_\ell$ be the algebra of differential polynomials
in the variables $u_\alpha$, $\alpha\in I=\{1,\dots,\ell\}$,
\begin{equation}\label{20170721:eq3}
\mc V_\ell=\mb F[u_\alpha^{(n)}\,|\,\alpha\in I,\,n\in\mb Z_{\geq0}]
\,.
\end{equation}
It is a differential algebra with respect to the derivation $\partial$
defined by
$\partial u_\alpha^{(n)}=u_\alpha^{(n+1)}$, $\alpha\in I,\,n\in\mb Z_{\geq0}$.
Note that the partial derivatives 
$\frac{\partial}{\partial u_\alpha^{(n)}}$, $\alpha\in I,\,n\in\mb Z_{\geq0}$,
are commuting derivations of $\mc V_\ell$ which satisfy
the following commutation relations:
\begin{equation}\label{20170718:eq2}
\bigg[\frac{\partial}{\partial u_\alpha^{(n)}},\partial\bigg]=\frac{\partial}{\partial u_\alpha^{(n-1)}}
\,\,,\,\,\,\,
\alpha\in I,\,n\in\mb Z_{\geq0}
\,,
\end{equation}
where the RHS is considered to be $0$ for $n=0$.

By the universal property of the algebra of differential polynomials,
for every collection of functions on space-time $f_\alpha\in\mc F$, $\alpha\in I$,
there exists a unique differential algebra homomorphism
$\mc V_\ell\rightarrow\mc F$,
mapping $u_\alpha^{(n)}\mapsto\partial^n f_\alpha$, $\alpha\in I,\,n\in\mb Z_{\geq0}$,
which we shall call \emph{evaluation at} $u=f$,
and we shall denote as
\begin{equation}\label{20170718:eq9b}
P=P(u,u^\prime,u^{\prime\prime},\dots)
\,\mapsto\,
P(f)=P(f,f^{\prime},f^{\prime\prime},\dots)
\,\in\mc F
\,.
\end{equation}

\begin{remark}\label{20170721:rem}
The results of the present paper can be generalized
to the case when $\mc V_\ell$ is replaced by an algebra of differential functions
extending the algebra of differential polynomials \eqref{20170721:eq3} \cite{BDSK09}.
On the other hand, for the KP hierarchy and all the hierarchies arising from classical affine
$\mc W$-algebras, the underlying differential algebra of unknown functions is
an algebra of differential polynomials
(in infinitely many variables for the KP hierarchy,
and finitely many variables for $\mc W$-algebras).
\end{remark}

\subsection{
Consistent linear systems of quasi-evolution equations
}
\label{sec:2.1b}
Let $\mc V=\mc V_\ell$ for simplicity of notation.
Recall that an \emph{evolutionary vector field} on $\mc V$
is a derivation $X:\,\mc V\to\mc V$ commuting with $\partial$.
It is immediate to see that all evolutionary vector fields 
on the algebra $\mc V$
are of the form
\begin{equation}\label{20170724:eq1}
X_P
=
\sum_{\alpha\in I}\sum_{n=0}^\infty (\partial^nP_\alpha)
\frac{\partial}{\partial u_\alpha^{(n)}}\,,
\end{equation}
for some $P=(P_\alpha)_{\alpha\in I}$, $P_\alpha\in\mc V$.
Hence, we have a bijective map from $\mc V^\ell$ to the space of evolutionary vector fields,
mapping $P$ to $X_P$.
Note that $[X_P,X_Q]=X_{[P,Q]}$, where
\begin{equation}\label{20210309:eq1}
[P,Q]=X_P(Q)-X_Q(P)\,.
\end{equation}
For $f\in\mc V$, we define its \emph{Frechet derivative} $D^{(f)}_\alpha(\partial)\in\mc V[\partial]$, $\alpha\in I$, as
\begin{equation}\label{eq:frechet}
D_\alpha^{(f)}(\partial)=\sum_{n=0}^{\infty}\frac{\partial f}{\partial u_\alpha^{(n)}}\partial^n\,.
\end{equation}
Obviously, for $P=(P_\alpha)_{\alpha\in I}\in\mc V^I$ and $f\in\mc V$, we have:
\begin{equation}\label{eq:frechet2}
X_P(f)=\sum_{\alpha\in I}D_\alpha^{(f)}(\partial)P_\alpha\,.
\end{equation}

Let $u=\{u_{\alpha}\}_{\alpha\in I}$ be the column vector of generators and let $P=(P_\alpha)_{\alpha\in I}\in\mc V^\ell$. By definition, an \emph{evolution equation} on $\mc V$ has the form
\begin{equation}\label{20170724:eq2}
\frac{\partial u}{\partial t}
=
P
\,.
\end{equation}
By the chain rule, $\frac{\partial}{\partial t}$ extends uniquely to the evolutionary vector field
$\frac{\partial}{\partial t}=X_P$ on $\mc V$.

We generalize the notion of evolution equation as follows:
a \emph{linear system of quasi-evolution equations} on $\mc V$ 
is a system of the form
\begin{equation}\label{20170724:eq3}
M(\partial)\frac{\partial u}{\partial t}
=
Q
\,,
\end{equation}
where $Q=(Q_j)_{j\in J}$, with $Q_j\in\mc V$, $j\in\{1,\dots,m\}$, and
 $M(\partial)$ is an $m\times \ell$ matrix differential operator over  $\mc V$. It is assumed that
$M(\partial)$ has only finitely many non-zero entries in each row if $\ell$ is infinite,
so that $M(\partial)$ defines a linear map $M(\partial):\mc V^\ell\to\mc V^m$.
Of course, an evolution equation is a special case of a linear system \eqref{20170724:eq3},
where  $M(\partial)=\id$ is the identity matrix.

\begin{definition}\label{def:consistency}
A linear system of evolution equations \eqref{20170724:eq3} is \emph{consistent}
if $Q\in\im M(\partial)\subset\mc V^m$.
\end{definition}
Suppose, for example, that  the matrix differential operator $M(\partial)$ 
has a right inverse $A(\partial)B(\partial)^{-1}$,
where $A(\partial)$ is an $\ell\times m$ matrix differential operator
and $B(\partial)$ is an invertible $m\times m$ matrix differential operator.
In other words $M(\partial)A(\partial)=B(\partial)$ is invertible as a pseudodifferential operator, i.e its inverse lies in
$\Mat_{m\times m}\mc V((\partial^{-1}))$.
If, moreover, the vector $Q=(Q_j)_{j\in J}$ lies in the image of $B(\partial)$,
then the system \eqref{20170724:eq3} is obviously consistent.

We also define a \emph{hierarchy of linear systems of quasi-evolution equations}
as a collection of linear systems,
\begin{equation}\label{20170724:eq6}
M_k(\partial)\frac{\partial u}{\partial t_k}
=
Q_{k}
\,\,,\,\,\,\,
Q_k\in\mc V^m
\,,
k\in\mc Z
\,,
\end{equation}
parametrized by a set $\mc Z\subset\mb Z_{\geq1}$. 
Suppose that each equation \eqref{20170724:eq6} of the hierarchy is consistent for every $k\in\mc Z$,
i.e. there exists $P_k\in\mc V^\ell$ such that
\begin{equation}\label{20170724:eq7}
M_k(\partial)P_k=Q_k
\,,\,\,
k\in\mc Z
\,.
\end{equation}
We also let $X_k=X_{P_k}:\mc V\to\mc V$ be the corresponding evolutionary vector fields as in \eqref{20170724:eq1}.
\begin{definition}\label{def:compatible}
The hierarchy \eqref{20170724:eq6} is \emph{compatible}
if we can choose elements $P_k\in\mc V^\ell$, $k\in\mc Z$, so that \eqref{20170724:eq7} holds
and
\begin{equation}\label{20170724:eq8}
X_k(P_{h})=X_h(P_{k})
\,\text{ for all }\,
h,k\in\mc Z
\,.
\end{equation}
\end{definition}
By equation \eqref{20210309:eq1}, this means that the evolutionary vector fields $X_k$, $k\in\mc Z$, commute:
$[X_k,X_h]=0$ for $k,h\in\mc Z$.
\begin{definition}\label{def:solution}
A \emph{solution} of the hierarchy of linear systems of quasi-evolution equations \eqref{20170724:eq6}
is a collection of functions $\varphi=\{\varphi_\alpha\mid\alpha\in I\}\subset\mc F$,
such that
\begin{equation}\label{20170724:eq5}
M_k(\varphi;\partial)\frac{\partial \varphi}{\partial t_k}
=Q_k(\varphi)
\,,
\quad k\in\mc Z
\,,
\end{equation}
where $M_{k}(\varphi;\partial)$ and $Q_{k}(\varphi)$
are obtained by applying the evaluation map \eqref{20170718:eq9b}.
\end{definition}
\begin{remark}\label{20170724:rem}
Note that the consistency of a linear system of quasi-evolution equations is NOT a necessary condition 
for the existence of solutions.
For example, the hierarchy of equations, corresponding to $L(\partial)=\partial^3+u$,
$$
0=0\frac{\partial u}{\partial t_1}=u'
\,\,,\,\,\,\,
\frac{\partial u}{\partial t_2}=0
\,,
$$
is not consistent, but it admits the solution $u(x,t_1,t_2)=1\in\mc F$.
\end{remark}

\subsection{
Pseudodifferential operators
}
\label{sec:2.5}

We will consider the algebra $\mc V((\partial^{-1}))$ 
of scalar pseudodifferential operators with coefficients in $\mc V$.
Given such an operator
\begin{equation}\label{20170719:eq1}
A(\partial)=\sum_{i=-\infty}^Na_i\partial^i
\,\,,\,\,\,\,
a_i\in\mc V
\,,
\end{equation}
its \emph{symbol} is defined as $A(z)=\sum_{i=-\infty}^Na_iz^{i}\in\mc V((z^{-1}))$.
Recall that the product of two pseudodifferential operators $A(\partial)$ and $B(\partial)$ is defined via their
symbols by
\begin{equation}\label{eq:product}
(AB)(z)=(A\circ B)(z)=A(z+\partial)(B(z))
\,.
\end{equation}
From \eqref{eq:product} we see that the subspace $\mc V[\partial]\subset\mc V((\partial^{-1}))$ of differential operators is a subalgebra.
Here and further,
a Laurent series involving negative powers of $z+\partial$ is always considered 
to be expanded using geometric series expansion in the domain of large $z$.

Let $A(\partial)$ be as in \eqref{20170719:eq1}. 
We denote 
by $A^*(\partial)=\sum_i(-\partial)^i\circ a_i\in\mc V((\partial^{-1}))$ its 
formal adjoint,
by $A(\partial)_+=\sum_{i=0}^Na_i\partial^i\in\mc V[\partial]$ 
its differential part,
by $A(\partial)_-=\sum_{i=-\infty}^{-1}a_i\partial^i\in\mc V[[\partial^{-1}]]\partial^{-1}$
its singular part.

The following notation will be used throughout the paper:
given $A(\partial)\in\mc V((\partial^{-1}))$ as in \eqref{20170719:eq1}
and $b,c\in \mc V$, we let:
\begin{equation}\label{eq:notation}
A(z+x)\big(\big|_{x=\partial}b)c
=\sum_{i=-\infty}^Na_i((z+\partial)^ib)c\,\in\mc V
\,.
\end{equation}
Using the notation \eqref{eq:notation}, we can rewrite
\begin{equation}\label{20210927:eq1}
A^*(z)=\big(\big|_{x=\partial}A(-z-x)\big)=\sum_{i=-\infty}^N(-z-\partial)^ia_i
\,,
\end{equation}
and the RHS of equation \eqref{eq:product} can be rewritten as
$A(z+x)(|_{x=\partial}B(z))$.

Furthermore, for $A(\partial)\in\mc V((\partial^{-1}))$ as in \eqref{20170719:eq1} (respectively, its symbol $A(z)\in\mc V((z^{-1}))$) we define its residue as
$\Res_\partial A(\partial)=a_{-1}$ (respectively, $\Res_zA(z)=a_{-1}$).
For a series involving negative powers of $z\pm w$ we shall use the notation $\iota_z$
or $\iota_w$ to denote geometric series expansion in the domain of large $z$
or of large $w$ respectively.
For example, 
$\iota_z(z-w)^{-1}=\sum_{n\in\mb Z_+}z^{-n-1}w^n$.
For $A(z)\in\mc V((z^{-1}))$ as above, we have
\begin{equation}\label{eq:positive}
\Res_z A(z)\iota_z(z-w)^{-1}
=
A(w)_+
\,\,,\quad
\Res_z A(z)\iota_w(z-w)^{-1}
=
-A(w)_-
\,.
\end{equation}

The \emph{order} of a pseudodifferential operator  $A(\partial)$ as in \eqref{20170719:eq1}
is $\ord(A)=N$ if $a_N\neq0$.
We also say that $A(\partial)$ is \emph{monic} if $a_N=1$.
Note that, if $A(\partial)$ and $B(\partial)\in\mc V((\partial))$ have orders $\ord(A)=N$ 
and $\ord(B)=M$ respectively,
then
\begin{equation}\label{eq:ord}
\ord(A\circ B)=M+N
\,\text{ and }\,
\ord([A,B])\leq M+N-1
\,,
\end{equation}
where $[A,B]$ denotes the commutator of the pseudodifferential operators $A$ and $B$.

The following results will be important in Section \ref{sec:3}.
\begin{lemma}{\cite[Lem.2.1a)]{DSKV16}}\label{lem:hn1a}
Given pseudodifferential operators $A(\partial),B(\partial)\in\mc V((\partial^{-1}))$, we have
$$
\Res_z A(z)B^*(-z)=\Res_z(A\circ B)(z)
\,.
$$
\end{lemma}
\begin{proof}
First, observe that, for an arbitrary $f(z)\in\mc V((z^{-1}))$, 
$\Res_zf(z+x)$, expanded in the domain $|z|>|x|$,
is independent of $x$.
Next, observe that replacing $z$ with $-z-x$ in equation \eqref{20210927:eq1}, we have the identity $A(z)=\left(\big|_{x=\partial}A^*(-z-x)\right)$.
It follows that
\begin{align*}
&\Res_z A(z)B^*(-z)
=
\Res_z A(z+x)\left(\big|_{x=\partial}B^*(-z-x)\right)
\\
&=
\Res_z A(z+\partial)B(z)
=
\Res_z (A\circ B)(z)
\,.
\end{align*}
\end{proof}
\begin{lemma}\label{20170726:lem}
Given pseudodifferential operators $A(\partial),B(\partial)\in\mc V((\partial^{-1}))$, 
the equation 
$\Res_z ((z+\partial)^nA(z))B^*(-z)=0$ holds for all $n\geq0$
if and only if
$(A\circ B)(\partial)_-=0$.
\end{lemma}
\begin{proof}
By Lemma \ref{lem:hn1a}, we have
\begin{equation}\label{20170726:eq2}
\begin{array}{l}
\displaystyle{
\vphantom{\Big(}
\Res_z ((z+\partial)^nA(z))B^*(-z)
=
\Res_{z}(z+\partial)^n A(z+\partial) B(z)
} \\
\displaystyle{
\vphantom{\Big(}
=
\Res_{z}(z+\partial)^n (A\circ B)(z)
=
\Res_zz^n(A\circ B)^*(-z)
\,.}
\end{array}
\end{equation}
For the first equality of \eqref{20170726:eq2} we applied Lemma \ref{lem:hn1a} to the pair 
of pseudodifferential operators $\partial^nA(\partial)$ and $B(\partial)$,
while for the last equality of \eqref{20170726:eq2} we applied Lemma \ref{lem:hn1a} to the pair 
$\partial^n$ and $A(\partial) B(\partial)$.
Obviously, the RHS of \eqref{20170726:eq2} is $0$ for every $n\geq0$ if and only if $(A\circ B)^*(-z)$
is a polynomial in $z$,
which is equivalent to saying that $(A\circ B)^*(\partial)_-=0$,
which in turn is equivalent to saying that $(A\circ B)(\partial)_-=0$.
\end{proof}
\begin{lemma}{\cite{Adl79},\cite[Lem.2.1b)]{DSKV16}}\label{lem:int-res}
For $A(\partial),B(\partial)\in\mc V((\partial^{-1}))$, we have
$$
\Res_\partial [A(\partial),B(\partial)]\in\partial\mc V
\,.
$$
\end{lemma}
\begin{proof}
It suffices to prove the claim for $A(\partial)=a\partial^i$, $B(\partial)=b\partial^j$,
with $a,b\in\mc V$ and $i,j\in\mb Z$.
We have
$$
\Res_\partial[a\partial^i,b\partial^j]
=
\binom{i}{i+j+1}ab^{(i+j+i)}-\binom{j}{i+j+1}ba^{(i+j+1)}
\,,
$$
where the binomial coefficient $\binom{p}{q}$ is intended to be $0$ for negative $q$.
Note that $ba^{(m)}\equiv(-1)^mab^{(m)}$ mod $\partial\mc V$.
The claim thus follows from the identity $\binom{j}{i+j+1}=(-1)^{i+j+1}\binom{i}{i+j+1}$.
\end{proof}

Given a collection of functions $f_\alpha\in\mc F$, $\alpha\in I$,
we have the homomorphism $\mc V((\partial^{-1}))\to\mc F((\partial^{-1}))$
extending the evaluation map \eqref{20170718:eq9b},
which we shall denote as
\begin{equation}\label{20170718:eq9c}
A=A(\partial)
\,\mapsto\,
A(f;\partial)
\,\in\mc F((\partial^{-1}))
\,.
\end{equation}
Of course all the results proved for the algebra $\mc V((\partial^{-1}))$ 
(including Lemmas \ref{lem:hn1a}, \ref{20170726:lem} and \ref{lem:int-res} above)
hold for the algebra $\mc F((\partial^{-1}))$.

\subsection{
$N$-th root of a monic pseudodifferential operator of order $N$
}
\label{sec:2.5c}

\begin{lemma}\label{20170720:lem}
Let $L(\partial)=\partial^N+a_0\partial^{N-1}+a_1\partial^{N-2}+\dots\in\mc V((\partial^{-1}))$
be a monic pseudodifferential operator of order $N\geq1$.
\begin{enumerate}[(a)]
\item
The inverse $L^{-1}(\partial)\in\mc V((\partial^{-1}))$
exists (and is unique) and it is a monic pseudodifferential operator of order $-N$.
\item
There exists a unique monic pseudodifferential operator of order $1$,
denoted $L^{\frac 1N}(\partial)$,
which is an $N$-th root of $L(\partial)$, i.e. such that 
\begin{equation}\label{20170720:eq1}
\big(L^{\frac1N}(\partial)\big)^N=L(\partial)
\,.
\end{equation}
\end{enumerate}
\end{lemma}
\begin{proof}
(See \cite{Dic03})
Part (a) is clear, since $L^{-1}(\partial)$ can be obtained by geometric series expansion.
For part (b), let
$$
L^{\frac1N}(\partial)=\partial+b_0+b_1\partial^{-1}+b_2\partial^{-2}+\dots
\,.
$$
The equation $(L^{\frac1N}(\partial))^N=L(\partial)$ translate, by looking at the coefficients of $\partial^{-i}$, 
to a system of equations of the form
$$
Nb_i+P_i=a_i
\,\,,\,\,\,\,
i\geq0
\,,
$$
where $P_0=0$ and, for $i\geq1$, $P_i$ is a polynomial with constant coefficients 
in the variables $b_j$, with $0\leq j<i$, 
and their derivatives.
(In fact, $P_i$ is homogeneous with respect to the grading defined by $\deg b_i^{(n)}=1+i+n$.)
Clearly, such system can be solved recursively and it admits a unique solution.
\end{proof}
\begin{lemma}\label{20170720:lem1}
Let $D_1$ and $D_2$ be derivations of $\mc V((\partial^{-1}))$
and let $L(\partial)$ be a monic pseudodifferential operator of order $N$.
If $D_1(L(\partial))=D_2(L(\partial))$, then $D_1(L^{\frac kN}(\partial))=D_2(L^{\frac kN}(\partial))$
for every $k\in\mb Z$.
\end{lemma}
\begin{proof}
By the Leibniz rule, it suffices to prove the claim for $k=1$.
We have
$$
\begin{array}{l}
\displaystyle{
\vphantom{\Big(}
0
=D_1(L(\partial))-D_2(L(\partial))
=D_1((L^{\frac1N}(\partial))^N)-D_2((L^{\frac1N}(\partial))^N)
} \\
\displaystyle{
\vphantom{\Big(}
=\sum_{h=0}^{N-1}L^{\frac hN}(\partial)\left(D_1(L^{\frac1N}(\partial))-D_2(L^{\frac1N}(\partial))\right)L^{\frac{N-h-1}N}(\partial)
\,.}
\end{array}
$$
Since $L^{\frac 1N}(\partial)$ is monic, the leading coefficient of the RHS is $N$ times the leading coefficient
of $D_1(L^{\frac1N}(\partial))-D_2(L^{\frac1N}(\partial))$, which must therefore be zero.
\end{proof}

\subsection{
Action of pseudodifferential operators on oscillating functions
}
\label{sec:2.6}

The algebra of differential operators $\mc V[\partial]$ acts naturally on $\mc V$,
and likewise the algebra of differential operators over $\mc F$ acts naturally on $\mc F$.
On the other hand, pseudodifferential operators 
do not act in any way on $\mc V$ or $\mc F$.

To overcome this problem suppose that $\mc F$ is an algebra of functions on space-time, with 
space variable $x$ and time variables $t_k$, $k\in\mc Z$. 
In this case, we define the \emph{space of oscillating functions} as
$$
\mc F((z^{-1})) e^{z\cdot\bm t}
=
\big\{
S(z)e^{z\cdot\bm t}\,\big|\, S(z)\in\mc F((z^{-1}))
\big\}
\,,
$$
where $e^{z\cdot\bm t}$ is just a formal symbol,
defined by the following rules
\begin{equation}\label{20210407:eq1}
\frac{\partial}{\partial x}e^{z\cdot\bm t}=ze^{z\cdot\bm t}
\,,\quad
\frac{\partial}{\partial t_k}e^{z\cdot\bm t}=z^ke^{z\cdot\bm t}\,,\quad k\in\mc Z\,,\quad
e^{z\cdot\bm t}|_{x=0,\bm t=0}=1\,.
\end{equation}
Namely, we should think of $z\cdot{\bm t}$ as $zx+\sum_{k\in\mc Z}z^kt_k$.

We have a natural representation of the algebra $\mc F((\partial^{-1}))$
on the space of oscillating functions, given by
$$
P(\partial)(S(z)e^{z\cdot\bm t})=(P\circ S)(z)e^{z\cdot\bm t}
\,,\,\,
\text{ for every }\,
P(\partial),S(\partial)\in\mc F((\partial^{-1}))
\,.
$$

On the space of oscillating functions we have a well-defined (commuting) action
of all time derivatives $\frac{\partial}{\partial t_k}$, $k\in\mc Z$,
induced by their action on $\mc F$ and by \eqref{20210407:eq1}:
\begin{equation}\label{20170725:eq5}
\frac{\partial}{\partial t_k}(S(z)e^{z\cdot\bm t})
=
\big(\frac{\partial S(z)}{\partial t_k}+z^kS(z)\big)e^{z\cdot\bm t}
\,.
\end{equation}

Recall that on the algebra of functions $\mc F$ we have the evaluation map 
$\mc F\twoheadrightarrow\mb F$ defined in \eqref{20170718:eq8b}.
It induces an \emph{evaluation map} on the space of oscillating functions
$\mc F((z^{-1}))e^{z\cdot\bm t}\twoheadrightarrow\mb F((z^{-1}))$,
defined by
\begin{equation}\label{20170725:eq6}
w(z)=S(z)e^{z\cdot\bm t}\mapsto w(z)|_{x=0,t=0}:=S(z)|_{x=0,t=0}
\,.
\end{equation}

\begin{lemma}\label{20170725:lem}
Let $w(z)=S(z)e^{z\cdot\bm t}$ be an oscillating function 
associated to a monic pseudodifferential operator 
$S(\partial)=\partial^N+s_0\partial^{N-1}+s_1\partial^{N-2}+\dots\in\mc F((\partial^{-1}))$.
If $P(\partial)\in\mc F((\partial^{-1}))$ is a pseudodifferential operator
such that $P(\partial)w(z)=0$, then $P(\partial)=0$.
\end{lemma}
\begin{proof}
It is a consequence of the obvious fact 
that a monic pseudodifferential operator  is not a zero divisor in $\mc F((\partial^{-1}))$.
\end{proof}

We also define the space of \emph{anti-oscillating functions} as
$\mc F((z^{-1}))e^{-z\cdot\bm t}$.
We define the action of $\mc F((\partial^{-1}))$ on the space of anti-oscillating functions
$\mc F((z^{-1}))e^{-z\cdot\bm t}$ by
$$
P(\partial)\big(S(-z)e^{-z\cdot\bm t}\big)
=
(P\circ S)(-z) e^{-z\cdot\bm t}
\,.
$$
Furthermore, we define the action of the partial derivatives $\frac{\partial}{\partial t_k}$, $k\in\mc Z$,
on the space of anti-oscillating functions $\mc F((z^{-1}))e^{-z\cdot\bm t}$ by
\begin{equation}\label{20170725:eq5b}
\frac{\partial}{\partial t_k}(S(z)e^{-z\cdot\bm t})
=
\big(\frac{\partial S(z)}{\partial t_k}-z^kS(z)\big)e^{-z\cdot\bm t}
\,.
\end{equation}

Given an oscillating function $w(z)=S(z)e^{z\cdot\bm t}$,
associated to an invertible pseudodifferential operator $S(\partial)\in\mc F((\partial^{-1}))$,
we define the corresponding \emph{adjoint} anti-oscillating function as
\begin{equation}\label{20170725:eq8}
w^\star(z)
:=
(S^*)^{-1}(-z)e^{-z\cdot\bm t}
\,.
\end{equation}

Note that two oscillating functions, or two anti-oscillating functions, cannot be multiplied.
On the other hand, we can multiply an oscillating function $w(z)=P(z)e^{z\cdot\bm t}$
and an anti-oscillating function $\omega(z)=Q(-z)e^{-z\cdot\bm t}$,
the result being a Laurent series:
\begin{equation}\label{20170726:eq1}
w(z)\omega(z)
=
P(z)Q(-z)\,\in\mc F((z^{-1}))
\,.
\end{equation}
It then also makes sense to take its residue $\Res_z(w(z)\omega(z))\in\mc F$,
which is, by definition, the coefficient of $z^{-1}$ of \eqref{20170726:eq1}.

\section{Lax operators and hierarchies of Lax equations}
\label{sec:3}

\subsection{Fractional powers}
\label{sec:3.1}

As in Section \ref{sec:2},
we let $\mc V$ be the algebra \eqref{20170721:eq3} of differential polynomials 
in $\ell$ variables $u_\alpha$, $\alpha\in I$.
Moreover, throughout this section
we fix a monic pseudodifferential operator $L(\partial)\in\mc V((\partial^{-1}))$ 
of order $N\geq1$:
\begin{equation}\label{eq:L}
L(\partial)=\partial^N+a_0\partial^{N-1}+a_1\partial^{N-2}+a_2\partial^{N-3}+\dots
\,\,,\,\,\,\,
a_i\in\mc V
\,.
\end{equation}
Recalling the definition \eqref{20170720:eq1} of the $N$-th root $L^{\frac1N}(\partial)$,
we let, for $k\in\mb Z$,
\begin{equation}\label{20170720:eq4}
B_k(\partial)=(L^{\frac kN}(\partial))_+\in\mc V[\partial]
\,\,,\,\,\,\,
B_k^{(-)}(\partial)=(L^{\frac kN}(\partial))_-\in\mc V[[\partial^{-1}]]\partial^{-1}
\,,
\end{equation}
so that
\begin{equation}\label{20170720:eq5}
L^{\frac kN}(\partial)=B_k(\partial)+B_k^{(-)}(\partial)
\,.
\end{equation}


Recall from \eqref{eq:frechet} the definition of the Frechet derivative $D_\alpha^{(f)}(\partial)\in\mc V[\partial]$, $\alpha\in I$,
associated to $f\in\mc V$.
It can be viewed as a differential operator
$$
D^{(f)}(\partial):\mc V^I\to\mc V\,,
$$
mapping $P=(P_\alpha)_{\alpha\in I}\mapsto\sum_{\alpha\in I}D_\alpha^{(f)}(\partial)P_\alpha$.
We generalize this map by replacing $f\in\mc V$ with an arbitrary pseudodifferential operator $L(\partial)\in\mc V((\partial^{-1}))$
as in \eqref{eq:L}. As a result we obtain a map
$$
D^{(L)}:\mc V^I\to\mc V((\partial^{-1}))\,,
$$
defined by taking the Frechet derivative of the coefficients $a_i\in\mc V$ of $L(\partial)$. In other words, recalling
\eqref{eq:frechet2}
$$
D^{(L)}(P)=\sum_{i=0}^\infty X_P(a_i)\partial^{N-1-i}
\,.
$$
We denote by $D^{(L)}(\mc V^I)$ the image of this map:
$$
D^{(L)}(\mc V^I)=\{D^{(L)}(P)|P\in\mc V^I\}\subset\mc V((\partial^{-1}))
\,.
$$

In the present section we review some results about scalar Lax operators, see e.g. \cite{GD76,DS85}. 
Let $L(\partial)$ be as in \eqref{eq:L}, let $B(\partial)\in\mc V[\partial]$, and consider the associated Lax equation
\begin{equation}\label{eq:lax0}
\frac{\partial L(\partial)}{\partial t_B}=[B(\partial),L(\partial)]\,.
\end{equation}
Recall that, if $\frac{\partial u}{\partial t_B}=P\in\mc V^I$, then,
by the chain rule, $L(\partial)$ evolves according to
$$
\frac{\partial L(\partial)}{\partial t_B}=\sum_{i=0}^{\infty}X_P(a_i)\partial^{N-1-i}=D^{(L)}(P)\in\mc V((\partial^{-1}))\,.
$$
Hence, equation \eqref{eq:lax0} can be rewritten as
\begin{equation}\label{eq:lax2}
D^{(L)}\left(\frac{\partial u}{\partial t_B}\right)=[B(\partial),L(\partial)]\,.
\end{equation}
We then obtain that equation \eqref{eq:lax0} is consistent (cf. Definition \ref{def:consistency}) if and only if
\begin{equation}\label{eq:A3}
[B(\partial),L(\partial)]=D^{(L)}(P) (=X_P(L(\partial)))\in\mc V((\partial^{-1}))\,,
\end{equation}
for some $P\in\mc V^I$.
With a slight abuse of notation, if \eqref{eq:A3} holds, we shall denote by $\frac{\partial}{\partial t_B}$ the (rather an)
associated evolutionary vector field $X_P$. The element $P=(P_\alpha)_{\alpha\in I}\in\mc V^I$ is defined up to adding an element in the kernel of $D^{(L)}$.

We want to describe more explicitly the space of operators $B(\partial)$ for which \eqref{eq:A3} holds.
Clearly, a necessary condition for \eqref{eq:A3} to hold is that
$\ord([B(\partial),L(\partial)])=\ord(X_P(L))\leq N-1$.
\begin{lemma}\label{lem:1}
For $B(\partial)\in\mc V[\partial]$ we have $\ord([B(\partial),L(\partial)])\leq N-1$ if and only if
$$
B(\partial)\in
\mc V+\Span_{\mb F}\{B_k(\partial)\}_{k\geq1}\subset\mc V[\partial]\,.
$$
In fact, 
\begin{equation}\label{eq:vic}
\ord([B_k(\partial),L(\partial)])\leq N-2
\,\,\text{ for every }\,
k\geq1\,.
\end{equation}
\end{lemma}
\begin{proof}
If $f\in\mc V$, by \eqref{eq:ord}, $[f,L(\partial)]$ has order at most $N-1$.
Obviously, $[L^{\frac{k}{N}}(\partial),L(\partial)]=0$ (since $L(\partial)=(L^{\frac1N}(\partial))^N$). Hence, by \eqref{20170720:eq5}, we have
$$
[B_k(\partial),L(\partial)]=-[B_k^{(-)}(\partial),L(\partial)]\,,
$$
which has order at most $N-2$. This proves the ``if'' part.

For the opposite implication, let $B(\partial)\in\mc V[\partial]$ be such that $\ord([B(\partial),L(\partial)])\leq N-1$. We shall prove,
by induction on $k=\ord (B(\partial))\geq0$, that $B(\partial)\in\mc V+\Span_{\mb F}\{B_h(\partial)\}_{h\geq1}$.
For $k=0$, the claim is obvious, so we can assume $k\geq1$. Let $a\in\mc V$ be the leading coefficient of $B(\partial)$:
$$
B(\partial)=a\partial^k+\text{ lower order terms}\,.
$$
We have
$$
[B(\partial),L(\partial)]=-Na'\partial^{N+k-1}+\text{ lower order terms}\,.
$$
Hence, $a'=0$, i.e. $a\in\mb F$. Since, by Lemma \ref{20170720:lem}(b), $L^{\frac1N}(\partial)$ (hence $B_k(\partial)$) is monic,
$$
\ord(B(\partial)-aB_k(\partial))\leq k-1\,,
$$
and, by the ``if'' part,
$$
\ord([B(\partial)-aB_k(\partial),L(\partial)])\leq N-1\,.
$$
The claim follows by the assumption of induction.
\end{proof}
Recall from Section \ref{sec:2.1}
the notion of \emph{compatibility}. We can study when two consistent Lax equations \eqref{eq:lax0} are compatible. This is discussed in the following proposition.
\begin{proposition}\label{20170719:prop}
Let $B(\partial),C(\partial)\in\mc V[\partial]$ be two differential operators satisfying the consistency condition \eqref{eq:A3}.
Then,
assuming that $B(\partial),C(\partial)\in\Span_{\mb F}\{B_k(\partial)\}_{k\geq1}$,
the corresponding Lax equations
\begin{equation}\label{20210329:eq*}
\frac{\partial L(\partial)}{\partial t_B}=[B(\partial),L(\partial)]\quad\text{and}\quad\frac{\partial L(\partial)}{\partial t_C}=[C(\partial),L(\partial)]
\end{equation}
are compatible on the subalgebra of $\mc V$ generated by the coefficients of $L(\partial)$.
In other words, the mixed second derivatives of $L(\partial)$, 
in virtue of \eqref{20210329:eq*}, coincide:
\begin{equation}\label{20170720:eq2}
\Big(\frac{\partial^2 L(\partial)}{\partial t_B\partial t_C}
=\Big)\,
\frac{\partial}{\partial t_B}
\big[C(\partial),L(\partial)\big]
=
\frac{\partial}{\partial t_C}
\big[B(\partial),L(\partial)\big]
\,.
\end{equation}
\end{proposition}
\begin{proof}
(See \cite{Dic03}.)
Recall that the Lax equations \eqref{20210329:eq*} are equivalent to linear systems of quasi-evolution equations as in \eqref{eq:lax2}.
Moreover,
consistency of \eqref{20210329:eq*} means that there exist $P_B,P_C\in\mc V^I$ such that
$[B(\partial),L(\partial)]=X_{P_B}(L(\partial))$ and $[C(\partial),L(\partial)]=X_{P_C}(L(\partial))$.
As before, by $\frac{\partial}{\partial t_B}$ and $\frac{\partial}{\partial t_C}$ we denote the evolutionary vector fields $X_{P_B}$
and $X_{P_C}$ respectively.
The compatibility condition \eqref{20170720:eq2} is then equivalent to
$$
[X_{P_B},X_{P_C}](L(\partial))=0\,.
$$
Since $\frac{\partial}{\partial t_B}$ and $[B(\partial),\,\cdot]$
are both derivations of $\mc V((\partial^{-1}))$
and $\frac{\partial L(\partial)}{\partial t_B}=[B(\partial),L(\partial)]$, 
by Lemma \ref{20170720:lem1} we have 
\begin{equation}\label{20170720:eq6}
\frac{\partial L^{\frac kN}(\partial)}{\partial t_B}
=
\big[B(\partial),L^{\frac kN}(\partial)\big]
\,,
\end{equation}
for every $k\geq0$.
By assumption, $B(\partial)=\sum_{k\in\mb Z_{\geq0}}\beta_kB_k(\partial)$, $C(\partial)=\sum_{k\in\mb Z_{\geq0}}\gamma_kB_k(\partial)$, 
with $\beta_k,\gamma_k\in\mb F$.
Let us denote $B^{(-)}(\partial)=\sum_k\beta_kB_k^{(-)}(\partial)$ and similarly for $C^{(-)}(\partial)$.
Hence, by the Leibniz rule,
\begin{equation}\label{20170720:eq3}
\begin{array}{l}
\displaystyle{
\vphantom{\Big(}
[X_{P_B},X_{P_C}](L(\partial))=\frac{\partial}{\partial t_B}
\big[C(\partial),L(\partial)\big]
-
\frac{\partial}{\partial t_C}
\big[B(\partial),L(\partial)\big]
} \\
\displaystyle{
\vphantom{\Big(}
=
\sum_{k\in\mb Z_{\geq0}}\left(
\gamma_k\Big[(\frac{\partial L^{\frac kN}(\partial)}{\partial t_B})_+,L(\partial)\Big]
-
\beta_k\Big[(\frac{\partial L^{\frac kN}(\partial)}{\partial t_C})_+,L(\partial)\Big]\right)
}
\\
\displaystyle{
\vphantom{\Big(}
+
\Big[C(\partial),\frac{\partial L(\partial)}{\partial t_B}\Big]
-
\Big[B(\partial),\frac{\partial L}{\partial t_C}\Big]
} \\
\displaystyle{
\vphantom{\Big(}
=
\sum_{k\in\mb Z_{\geq0}}\left(
\gamma_k\big[\big[B(\partial),L^{\frac kN}(\partial)\big]_+,L(\partial)\big]
+\beta_k\big[\big[L^{\frac kN}(\partial),C(\partial)\big]_+,L(\partial)\big]
\right)
} \\
\displaystyle{
\vphantom{\Big(}
+\big[C(\partial),\big[B(\partial),L(\partial)\big]\big]
-\big[B(\partial),\big[C(\partial),L(\partial)\big]\big]}
\\
\displaystyle{
\vphantom{\Big(}
=
\big[\big[B(\partial),C(\partial)+C^{(-)}(\partial)\big]_+,L(\partial)\big]
+\big[\big[B(\partial)+B^{(-)}(\partial),C(\partial)\big]_+,L(\partial)\big]
}
\\
\displaystyle{
\vphantom{\Big(}
-\big[\big[B(\partial),C(\partial)\big],L(\partial)\big]
\,.}
\end{array}
\end{equation}
In the second equality we used the obvious fact that applying $\frac{\partial}{\partial t}$
to an element of $\mc V((\partial^{-1}))$ 
commutes with taking its positive part,
and in the third equality we used the Lax equations \eqref{20210329:eq*} and \eqref{20170720:eq6}.
Note that 
$$
\begin{array}{l}
\displaystyle{
\vphantom{\Big(}
\big[B(\partial)+B^{(-)}(\partial),C(\partial)\big]_+
-\big[B(\partial),C(\partial)]
}
\\
\displaystyle{
\vphantom{\Big(}
=
\big[B(\partial)+B^{(-)}(\partial),C(\partial)\big]_+
-\big[B(\partial),C(\partial)\big]_+
} \\
\displaystyle{
\vphantom{\Big(}
=
\big[B^{(-)}(\partial),C(\partial)\big]_+
=
\big[B^{(-)}(\partial),C(\partial)+C^{(-)}(\partial)\big]_+
\,.}
\end{array}
$$
Hence, the RHS of \eqref{20170720:eq3} becomes
$$
\begin{array}{l}
\displaystyle{
\vphantom{\Big(}
\big[\big[B(\partial)+B^{(-)}(\partial),C(\partial)+C^{(-)}(\partial)\big]_+,L(\partial)\big]
}\\
\displaystyle{
\vphantom{\Big(}
=\sum_{h,k\in\mb Z_{\geq0}}\beta_h\gamma_k
\big[\big[L^{\frac hN}(\partial),L^{\frac kN}(\partial)\big]_+,L(\partial)\big]
=0
\,.}
\end{array}
$$
\end{proof}

The Lax equation \eqref{eq:lax0} is useful since it admits a large family of integrals of motion.
For $h\in\mc V$ we denote $\tint h$ its image in the quotient space $\mc V/\partial\mc V$.
Recall that an \emph{integral of motion} for an evolution equation
is an element $\tint h\in\mc V/\partial\mc V$ such that $\frac{d}{dt}\tint h=0$.
The following fact is well known (see e.g. \cite{Adl79,Dic03}).
\begin{proposition}\label{prop:int-lax}
Assume that the Lax equation \eqref{eq:lax0} is consistent, for a pseudodifferential operator $L(\partial)$, monic of order $N\geq1$,
and a differential operator $B(\partial)$.
Then, $\tint\Res_\partial L^{\frac nN}(\partial)$ is an integral of motion of \eqref{eq:lax0} for every $n\geq0$.
\end{proposition}
\begin{proof}
It immediately follows from \eqref{20170720:eq6} and Lemma \ref{lem:int-res}.
\end{proof}

\subsection{Lax operators}
\label{sec:3.2}

In view of Lemma \ref{lem:1} and Proposition \ref{20170719:prop} we introduce the following:
\begin{definition}\label{def:lax}
A (scalar) \emph{Lax operator} is a pseudodifferential operator $L(\partial)\in\mc V((\partial^{-1}))$ such that 
the \emph{Lax equation}
\begin{equation}\label{eq:lax}
\frac{\partial L(\partial)}{\partial t_k}=[B_k(\partial),L(\partial)]\,,
\end{equation}
is consistent and non-zero (i.e. the RHS is non-zero) for infinitely may values of $k\geq1$.
We denote by $\mc Z_L\subset\mb Z_{\geq1}$ the (infinite) set of $k$'s such that equation \eqref{eq:lax}
is consistent.
\end{definition}
As an immediate consequence of Lemma \ref{lem:1} and Proposition \ref{20170719:prop} we have the following result.
\begin{corollary}
For a scalar Lax operator $L(\partial)\in\mc V((\partial^{-1}))$ all equations of the hierarchy \eqref{eq:lax},
with $k\in\mc Z_L$, are consistent and are compatible on the differential subalgebra of $\mc V$ 
generated by the coefficients of $L(\partial)$.
\end{corollary}

\subsection{Examples of Lax operators}
\label{sec:3.3}

Let $L(\partial)$ be a monic pseudodifferential operator as in \eqref{eq:L}.
By \eqref{eq:vic},
the Lax equation \eqref{eq:lax} implies $\frac{\partial a_0}{\partial t_k}=0$.
Hence, without loss of generality, we may assume that $a_0=0$.
In this case $B_1(\partial)=\partial$, hence \eqref{eq:lax} for $k=1$
is the trivial equation $\frac{\partial a_j}{\partial t_1}=\partial a_j$.
\begin{example}\label{ex:1}
Let $\mc V$ be the algebra of differential polynomials in infinitely many differential variables $u_i$, $i\in\mb Z_{\geq1}$.
Let $L(\partial)$ be the Sato operator, i.e. the pseudodifferential operator \eqref{eq:L}, with $a_0=0$ and $a_i=u_i$ for all $i\geq1$.
It follows from \eqref{eq:vic}
that the Lax equation \eqref{eq:lax} is consistent for all $k\geq1$.
So, in this case $\mc Z_L=\mb Z_{\geq1}$ and all Lax equations are non-zero, hence $L(\partial)$ is a Lax operator.
The corresponding hierarchy of (compatible) Lax equations is called the \emph{KP hierarchy}.
Note that this hierarchy is isomorphic to that with $N=1$ \cite{DSKV15}.
As will be shown in Example \ref{exa:KP},
from the Lax equation one derives, in this case, the classical KP equation.
\end{example}
\begin{example}\label{ex:2}
Let $\mc V$ be the algebra of differential polynomials in $u_1,\dots,u_{N-1}$, where $N\geq2$,
and let $L(\partial)=\partial^N+u_1\partial^{N-2}+\dots+u_{N-1}$,
which is called the Gelfand-Dickey operator.
Again thanks to \eqref{eq:vic}, all Lax equations \eqref{eq:lax} are consistent and they are nonzero if $k$ is not divisible by $N$.
Hence $\mc Z_L=\mb Z_{\geq1}$ and $L(\partial)$ is a Lax operator.
The corresponding hierarchy of evolution equation is called the $N$-th Gelfand-Dickey or $N$-th KdV hierarchy.
For $N=2$, $\mc V$ is an algebra of differential polynomials in a single variable $u=u_1$, 
the Lax operator is $L(\partial)=\partial^2+u$,
and the first non-trivial Lax equation is $\frac{\partial L(\partial)}{\partial t_3}=[B_3(\partial),L(\partial)]$,
which is the KdV equation $\frac{\partial u}{\partial t_3}=\frac14 u'''+\frac32 uu'$.
For $N=3$, $\mc V$ is an algebra of differential polynomials in two variables $u=u_1,v=u_2$ and 
$L(\partial)=\partial^3+u\partial+v$.
One can easily get $B_2(\partial)=\partial^2+\frac23u$.
Hence, the Lax equation \eqref{eq:lax} for $k=2$
becomes the Boussinesq equation
$$
\frac{\partial u}{\partial t_2}=2v'-u''
\,\,\,,\,\,\,
\frac{\partial v}{\partial t_2}=v''-\frac{2}{3}u'''-\frac23 uu'
\,.
$$
However, one can show that $L(\partial)=\partial^3+u$ is not a Lax operator.
\end{example}
\begin{example}\label{ex:3}
Let $\mc V$ be the algebra of differential polynomials in the variables $u_1,\dots,u_{n-1},p_1,q_1,\dots,p_s,q_s$,
and consider the pseudodifferential operator
$$
L(\partial)=\partial^n+u_1\partial^{n-2}+\dots+u_{n-1}+\sum_{i=1}^sp_i\partial^{-1}\circ q_i
\,.
$$
This is a Lax operator with $\mc Z_L=\mb Z_{\geq1}$ \cite{Che92,KSS91,KS92,SS93}.
The corresponding Lax hierarchy \eqref{eq:lax} is called the $s$-vector $n$-constrained KP hierarchy.
It is obtained by a certain reduction of the KP hierarchy,
and it follows from \cite{DSKV16a,CDSKVvdL20} that it is compatible.
The simplest special case $L(\partial)=\partial+p\partial^{-1}\circ q$
gives the famous non-linear Schroedinger hierarchy,
while $L(\partial)=\partial^2+u+p\partial^{-1}\circ q$ gives the Yajima-Oikawa hierarchy \cite{YO76}.
\end{example}
\begin{example}\label{ex:4}
A generalization of both Examples \ref{ex:2} and \ref{ex:3}
is considered in \cite{DSKV16a}.
Namely, for each partition $\underline p$ of $N$
we construct an $r\times r$ matrix Lax operator $L(\partial)$,
where $r$ is the multiplicity of the largest part in $\underline p$,
for which $\mc Z_L=\mb Z_{\geq1}$.
The special case of the partition $N=N$ corresponds 
to the $N$-th Gelfand Dickey hierarchy of Example \ref{ex:2},
while the partition $N=n+1+\dots+1$ (with $s=N-n$ ones)
corresponds to the $s$-vector $n$-constrained KP hierarchy in Example \ref{ex:3}.
This is obtained by making use of the classical $\mc W$-algebra
associated to the nilpotent element of $\mf{sl}_N$ with Jordan form
corresponding to the partition $\underline{p}$,
generalizing the construction of \cite{DS85}.
\end{example}
\begin{remark}\label{rem:so_sp}
The $r\times r$ matrix Lax operator $L(\partial)$ in Example \ref{ex:4} is obtained as the quasideterminant (cf. \cite{DSKV16a})
$$
L(\partial)=(JA(\partial)^{-1}I)^{-1}
$$
where $f^{p_{1}-1}=JI$ is the ``canonical decomposition'' of the $p_1-1$-th power of the nilpotent element $f$
of $\mf{sl}_N$ with Jordan form
corresponding to the partition $\underline{p}$ and $A(\partial)$ is a certain first order matrix differential operator.
For $\mf g=\mf {so}_N$ or $\mf{sp}_N$ we have that $A^{*}(\partial)=-A(\partial)$ (cf. \cite{DSKV18}).
Hence, using Lemma 5.5 in \cite{DSKV18} (with $T=f^{p_1-1}$) we get that the Lax operator for the integrable hierarchy corresponding to the classical $\mc W$-algebra
associated to $\mf g$ and the partition $\underline p$ satisfies the further condition
$$
L^*(\partial)=(-1)^{p_1}L(\partial)\,,
$$
where the adjoint is computed with respect to a bilinear form on an $r$-dimensional vector space which has
parity $\epsilon (-1)^{p_1-1}$, with $\epsilon=1$ for $\mf{so}_N$ and $\epsilon=-1$ for $\mf {sp}_N$.
In this case we have $\mc Z_L=2\mb Z_{\geq0}+1$.
\end{remark}
\begin{remark}\label{rem:odd}
Let $L(\partial)$ be a monic pseudodifferential operator of order $N$
which is $(-1)^N$-adjoint, i.e. selfadjoint for even $N$ and skewadjoint for odd $N$.
Then by uniqueness of the $N$-th root of a monic operator,
$L^{\frac1N}(\partial)$ is necessarily a skewadjoint operator.
Therefore, $L^{\frac kN}(\partial)$ is $(-1)^k$-adjoint,
hence $B_k(\partial)$ is $(-1)^k$-adjoint as well.
Therefore $[B_k(\partial),L(\partial)]$ is $(-1)^{N+k+1}$-adjoint.
And therefore the corresponding Lax equation \eqref{eq:lax} can be consistent only for odd $k$.
\end{remark}
\begin{example}\label{ex:5}
Let $\mc V$ be the algebra of differential polynomials in the variables $v_j$ with $j\geq1$ odd,
and let $L(\partial)$ be the ``generic'' monic skewadjoint pseudodifferential operator of order $1$, namely
\begin{equation}\label{eq:312a}
L(\partial)=\partial+\frac12\sum_{j\geq1,\,\text{odd}}(\partial^{-j}\circ v_j+v_j\partial^{-j})
\,.
\end{equation}
Then, by \eqref{eq:vic}, $[B_k(\partial),L(\partial)]$ has negative order, and by Remark \ref{rem:odd} it is skewadjoint for every odd $k$.
Hence, since $L(\partial)$ is generic, the corresponding Lax equation \eqref{eq:lax} is consistent and non-zero for every odd $k$,
so that $\mc Z_L=2\mb Z_{\geq0}+1$ and \eqref{eq:312a} is a Lax operator.
The corresponding hierarchy of Lax equation is called the CKP hierarchy \cite{DJKM81}.
Note that as in the KP case we could have taken instead $L(\partial)$ be a generic monic $(-1)^N$-adjoint pseudodifferential 
operator of order $N$, again obtaining the CKP hierarchy.
\end{example}
\begin{remark}\label{rem:odd-d}
Let $L(\partial)$ be a monic pseudodifferential operator of order $N$
satisfying
\begin{equation}\label{eq:sesquiadj}
L^*(\partial)=(-1)^N\partial\circ L(\partial)\partial^{-1}
\,.
\end{equation}
We call such an operator $(-1)^N$-\emph{sesquiadjoint}.
Observe that any such operator can be written as
\begin{equation}\label{eq:sesquiadj1}
L(\partial)
=
\partial^N
+
\frac12\sum_{j\in\mb Z_{\geq1}}(a_j\partial^{N-2j}+\partial^{N-2j-1}\circ a_j\partial)
\,,
\end{equation}
with $a_j\in\mc V$ for every $j$.
Observe also that, if $N$ is odd, then
\begin{equation}\label{eq:sesquiadj2}
\Res_\partial L(\partial)\partial^{-1}=0
\,.
\end{equation}
Indeed, since $N$ is odd, neither $a_j\partial^{N-2j}$ nor $\partial^{N-2j-1}\circ a_j\partial$ 
can contribute to the residue \eqref{eq:sesquiadj2} for any $j\in\mb Z_{\geq0}$.
It is immediate to check, by the uniqueness of the $N$-th root of a monic operator,
that $L^{\frac1N}(\partial)$ is automatically $(-1)$-sesquiadjoint
and, therefore, $L^{\frac kN}(\partial)$ is $(-1)^k$-sesquiadjoint.
As a consequence, for odd $k$, 
$$
B_k(\partial)^*
=
\big((L^{\frac kN}(\partial))_+\big)^*
=
\big((L^{\frac kN}(\partial))^*\big)_+
=
(-1)^k\big(\partial\circ L^{\frac kN}(\partial)\partial^{-1}\big)_+
=
-\partial\circ B_k(\partial)\partial^{-1}
$$
since, by \eqref{eq:sesquiadj2}, $L^{\frac kN}(\partial)$ does not have constant term.
Hence, $B_k(\partial)$ is $(-1)^k$-sesquiadjoint as well.
It follows that $[B_k(\partial),L(\partial)]$ is $(-1)^{N}$-sesquiadjoint only for odd $k$,
so the corresponding Lax equation \eqref{eq:lax} can be consistent only for odd $k$.
\end{remark}
\begin{example}\label{ex:6}
Let $\mc V$ be the same as in Example \ref{ex:5}
and let $L(\partial)$ be the ``generic'' monic pseudodifferential operator of order $1$ which is $(-1)$-sesquiadjoint,
as defined in \eqref{eq:sesquiadj}.
Namely, by \eqref{eq:sesquiadj1},
\begin{equation}\label{eq:312b}
L(\partial)=\partial+\frac12\sum_{j\geq1,\,\text{odd}}(v_j\partial^{-j}+\partial^{-j-1}\circ v_j\partial)
\,.
\end{equation}
Then by \eqref{eq:vic} $[B_k(\partial),L(\partial)]$ has negative order and by Remark \ref{rem:odd-d} it is $(-1)$-sesquiadjoint for every odd $k$.
Hence, since $L(\partial)$ is generic,  the corresponding Lax equation \eqref{eq:lax} is consistent and non-zero for every odd $k$,
so that $\mc Z_L=2\mb Z_{\geq0}+1$ and \eqref{eq:312b} is a Lax operator.
The corresponding hierarchy of Lax equations is called the BKP hierarchy \cite{DJKM81}.
As in Example \ref{ex:5} we could have taken $L(\partial)$ of order $N$, obtaining again the BKP hierarchy.
\end{example}
\begin{example}\label{ex:7}
Let $L(\partial)$ be the ``generic'' monic $(-1)^N$-adjoint differential operator of order $N$, namely
\begin{equation}\label{eq:315}
L(\partial)=\partial^N+\frac12\sum_{j=1}^{\left[\frac N2\right]}(\partial^{N-2j}\circ v_j+v_j\partial^{N-2j})
\,,
\end{equation}
where $\{v_j\}_{j=1}^{\left[\frac N2\right]}$ are the differential variables generating the algebra of differential polynomials $\mc V$.
Then, by \eqref{eq:vic}, $[B_k(\partial),L(\partial)]$ is a differential operator of order bounded by $N-2$, 
and by Remark \ref{rem:odd} it is $(-1)^N$-adjoint for every odd $k$.
Hence, since $L(\partial)$ is generic, the corresponding Lax equation \eqref{eq:lax} is consistent for every odd $k$
and it is non-zero provided that $k$ is not divisible by $N$,
so that $\mc Z_L=2\mb Z_{\geq0}+1$ and \eqref{eq:315} is a Lax operator.
The corresponding hierarchy of Lax equations is called the Drinfeld-Sokolov hierarchy associated 
to $(A_{N-1}^{(2)},c_0)$ for odd $N$ and to $\mf{sp}_N$ for even $N$, \cite{DS85}.
The simplest case $N=2$ gives the KdV hierarchy,
while $N=3$ gives the Kaup-Kupershmidt operator $L(\partial)=\partial^3+u\partial+\frac12 u'$,
and the Lax equation \eqref{eq:lax} with $k=5$ is the Kaup-Kupershmidt equation \cite{Kaup80}.
\end{example}
\begin{example}\label{ex:8}
Let $L(\partial)$ be the ``generic'' monic $(-1)^N$-sesquiadjoint differential operator of order $N$, namely
\begin{equation}\label{eq:315a}
L(\partial)=\partial^N+\frac12\sum_{j=1}^{\left[\frac{N-1}2\right]} (v_j\partial^{N-2j}+\partial^{N-2j-1}\circ v_j\partial)
\,,
\end{equation}
where $\{v_j\}_{j=1}^{\left[\frac{N-1}2\right]}$ are the differential variables generating the algebra of differential polynomials $\mc V$.
Then, by \eqref{eq:vic}, $[B_k(\partial),L(\partial)]$ is a differential operator of order bounded by $N-2$, 
and by Remark \ref{rem:odd} it is $(-1)^N$-sesquiadjoint for every odd $k$.
Hence, since $L$ is generic, the corresponding Lax equation \eqref{eq:lax} is consistent for every odd $k$
and it is non-zero provided that $k$ is not divisible by $N$,
so that $\mc Z_L=2\mb Z_{\geq0}+1$ and \eqref{eq:315a} is a Lax operator.
The corresponding hierarchy of Lax equations is called the Drinfeld-Sokolov hierarchy associated to $(A_{N-1}^{(2)},c_{\frac{N-1}{2}})$ for odd $N$ 
and associated to $(D_{\frac{N}{2}}^{(2)},c_0)$ for even $N$.
The simplest case $N=3$ gives $L(\partial)=\partial^3+u\partial$
and the Lax equation \eqref{eq:lax} for $k=5$ is the Sawada-Kotera equation \cite{SK74}.
\end{example}
\begin{example}\label{ex:9}
For odd $N$ and any $s\geq1$, let 
$$
L(\partial)=\partial^{N}+\frac12\sum_{j=1}^{\frac{N-1}2}(\partial^{N-2j}\circ v_j+v_j\partial^{N-2j})
+\sum_{i=1}^sp_{i}\partial^{-1}\circ p_{i}
\,,
$$
where $v_1,\dots,v_{\frac {N-1}2},p_1,\dots,p_s$ are the differential variables generating the algebra of differential polynomials $\mc V$.
For $s=1$ it is a Lax operator with $\mc Z_L=2\mb Z_{\geq0}+1$,
since the corresponding hierarchy of Lax equations is the Drinfeld-Sokolov hierarchy associated 
to $(A_{N}^{(2)},c_{\frac{N-1}{2}})$.
The simplest case $N=2$ gives the operator $L(\partial)=\partial+u\partial^{-1}\circ u$
and the corresponding hierarchy is the modified KdV hierarchy.
For arbitrary $s\geq1$ 
the pseudodifferential operator $L(\partial)$ can be characterized as the ``generic'' skewadjoint pseudodifferential operator of the form given in Example \ref{ex:3}. Hence, by Remark \ref{rem:so_sp}, this is the Lax operator for the integrable hierarchy 
(after a suitable Dirac reduction) constructed
in \cite{DSKV18} for the classical affine $\mc W$-algebra associated to $\mf{so}_{N+s}$ and a nilpotent element $f$ 
with Jordan form corresponding to the partition $\underline p=(N,1^s)$.
\end{example}
\begin{example}\label{ex:9b}
For even $N$ and even $s\geq1$, let 
$$
L(\partial)=\partial^{N}+\frac12\sum_{j=1}^{\frac{N}2}(\partial^{N-2j}\circ v_j+v_j\partial^{N-2j})
+\sum_{i=1}^{\frac {s}{2}}\big(p_{i}\partial^{-1}\circ q_{i}-q_i\partial^{-1}\circ p_i\big)
\,,
$$
where $v_1,\dots,v_{\frac {N}2},p_1,q_1,\dots,p_{\frac s2},q_{\frac{s}{2}}$ are the differential variables generating the algebra of differential polynomials $\mc V$.
As in Example \ref{ex:9}, by Remark \ref{rem:so_sp}, it follows that it is a Lax operator (with $\mc Z_L=2\mb Z_{\geq0}+1)$
for the integrable hierarchy (after a suitable Dirac reduction) constructed
in \cite{DSKV18} for the classical affine $\mc W$-algebra associated to $sp_{N+s}$ and a nilpotent element $f$ whose Jordan form
corresponds to the partition $\underline p=(N,1^s)$. In fact, $L(\partial)$ can be characterized as the ``generic'' selfadjoint pseudodifferential operator of the form given in Example \ref{ex:3}.
\end{example}
\begin{example}\label{ex:KN}
Consider the ``generic'' selfadjoint differential operator of order $4$:
\begin{equation}\label{eq:kn1}
L(\partial)=\partial^4+\partial\circ a\partial+b
\,,
\end{equation}
over the algebra $\mc V$ of differential polynomials in $a$ and $b$.
Its $4$-th root has the form
\begin{equation}\label{eq:kn2}
L^{\frac14}(\partial)=\partial+\sum_{j=1}^\infty v_j\partial^{-j}
\,,
\end{equation}
where the coefficients $v_j\in\mc V$ satisfy a certain explicit recurrence relation,
obtained from the identity $(L^{\frac14}(\partial))^4=L(\partial)$.
In particular, 
\begin{equation}\label{eq:kn3}
v_1=\frac{a}4
\,,\,\,
v_2=-\frac{a'}8
\,.
\end{equation}
We want to write explicitly the Lax equation \eqref{eq:lax} associated to $L(\partial)$ for $k=3$.
By \eqref{eq:kn2} and \eqref{eq:kn3} we get
$$
B_3(\partial)
=
(L^{\frac34}(\partial))_+
=
\partial^3+3v_1\partial+3(v_2+v_1')
=
\partial^3+\frac34a\partial+\frac38 a'
\,.
$$
Hence, by direct computation,
$$
[B_3(\partial),L(\partial)]
=
\partial\circ P\partial+Q
\,,
$$
where
$$
P=3b'-\frac54a'''-\frac34aa'
\,\,,\,\,\,\,
Q=b'''-\frac38a^{(5)}-\frac38(aa'''+a'a'')+\frac34ab'
\,.
$$
As a consequence, the Lax equation \eqref{eq:lax} for $k=3$ and the operator \eqref{eq:kn1} 
is the following system of evolution equations on $a$ and $b$:
\begin{equation}\label{eq:kn4}
\frac{\partial a}{\partial t}
=
3b'-\frac54a'''-\frac34aa'
\,\,,\,\,\,\,
\frac{\partial b}{\partial t}
=
b'''-\frac38a^{(5)}-\frac38(aa'''+a'a'')+\frac34ab'
\,.
\end{equation}
Next, we notice that the system \eqref{eq:kn4} remains consistent once we impose the constraint
\begin{equation}\label{eq:kn5}
b=\frac12 a''+\frac14a^2
\,.
\end{equation}
Indeed, as one can easily check,
$$
\frac{\partial b}{\partial t}
=
Q
=
\frac12P''+\frac12aP
=
\frac{\partial}{\partial t}\Big(\frac12 a''+\frac14a^2\Big)
\,,
$$
provided that \eqref{eq:kn5} holds.
As a consequence, the system \eqref{eq:kn4} can be greatly simplified with the change of variables $(a,b)\to(a,u)$,
where
\begin{equation}\label{eq:kn6}
b=\frac12 a''+\frac14a^2-\frac{c}3u
\,,
\end{equation}
and $c$ is an arbitrary non-zero constant.
Indeed, equations \eqref{eq:kn4} can be rewritten in terms of these new variables as
\begin{equation}\label{eq:kn7}
\frac{\partial u}{\partial t}
=
-\frac12u'''-\frac34au'
\,\,,\,\,\,\,
\frac{\partial a}{\partial t}
=
\frac14a'''+\frac34aa'-cu'
\,.
\end{equation}
Again we can ask under which constraint relating $a$ and $u$ the system \eqref{eq:kn7} remains consistent.
A solution to this question is given by the following constraint (see \cite{Sok84}):
\begin{equation}\label{eq:kn8}
a=-\frac{u'''}{u'}+\frac12\frac{(u'')^2}{(u')^2}-\frac13\frac{h(u)}{(u')^2}
\,,
\end{equation}
where $h$ is an arbitrary polynomial of degree $3$ in $u$ and the leading coefficient $c$.
In this case the second equation of \eqref{eq:kn7} remains consistent with the first, which reduces to the 
famous Krichever-Novikov equation \cite{KN80}
\begin{equation}\label{eq:kn9}
4\frac{\partial u}{\partial t}
=
u'''-\frac32\frac{(u'')^2}{u'}+\frac{h(u)}{u'}
\,.
\end{equation}
Here consistency means that applying $D_a(\partial)$, for $a$ as in \eqref{eq:kn8},
to the RHS of the first equation in \eqref{eq:kn7}, we obtain the RHS of the second equation.

Consider the differential operator $L(\partial)$ obtained from \eqref{eq:kn1}
by substituting $a$ and $b$ as \eqref{eq:kn6} and \eqref{eq:kn8},
obtaining an operator $L(\partial)$ of order $4$ discovered by Sokolov \cite{Sok84},
who showed that \eqref{eq:kn9} is equivalent to the Lax equation $\frac{\partial L(\partial)}{\partial t_3}=[B_3(\partial),L(\partial)]$.
One may expect that all equations $\frac{\partial L(\partial)}{\partial t_k}=[B_k(\partial),L(\partial)]$ for odd $k$ are consistent,
but this is still an open problem.
\end{example}

\section{Lax equations and Zakharov-Shabat equations}
\label{sec:4}

Let $L(\partial)\in\mc V((\partial^{-1}))$ be as in \eqref{eq:L}.
In the terminology of Section \ref{sec:2.1}
we can talk about the \emph{solutions} of the hierarchy \eqref{eq:lax},
as a collection of functions on space-time $\varphi_\alpha\in\mc F$, $\alpha\in I$,
such that 
\begin{equation}\label{eq:lax-sol}
\frac{\partial L(\varphi;\partial)}{\partial t_k}
=
[B_k(\varphi;\partial),L(\varphi;\partial)]
\,\quad
k\in\mc Z_L
\,,
\end{equation}
where $L(\varphi;\partial)$ and $B_k(\varphi;\partial)$
are obtained applying the evaluation map $\mc V((\partial^{-1}))\to\mc F((\partial^{-1}))$
defined in \eqref{20170718:eq9c}
to $L(\partial)$ and $B_k(\partial)$ respectively.
In \eqref{eq:lax-sol} we assume that the functions on space-time $\varphi\in\mc F$ depend on time variables $t_k$, $k\in\mc Z_L$. 

\begin{proposition}\label{20170720:prop}
\begin{enumerate}[(a)]
\item
Let $L(\partial)$ be a Lax operator, let $h,k\in\mc Z_L$, and let $\frac{\partial}{\partial t_k},\frac{\partial}{\partial t_h}:\mc V\to\mc V$
denote (as usual) the associated evolutionary vector fields. 
Then the following \emph{Zakharov-Shabat} equations ($=$ \emph{zero-curvature} equations)
hold
\begin{equation}\label{eq:zs}
\frac{\partial B_k(\partial)}{\partial t_h}-\frac{\partial B_h(\partial)}{\partial t_k}
+[B_k(\partial),B_h(\partial)]=0
\,,\,\,
h,k\in\mc Z_L
\,,
\end{equation}
as well as the following \emph{complementary Zakharov-Shabat} equations
\begin{equation}\label{eq:czs}
\frac{\partial B_k^{(-)}(\partial)}{\partial t_h}-\frac{\partial B_h^{(-)}(\partial)}{\partial t_k}
-[B_k^{(-)}(\partial),B_h^{(-)}(\partial)]=0
\,,\,\,
h,k\in\mc Z_L
\,,
\end{equation}
where $B_k^{(-)}(\partial)$ is as in \eqref{20170720:eq4}.
\item
Conversely, suppose that there exist infinitely many non-zero
evolutionary vector fields $\frac{\partial}{\partial t_k}:\mc V\to\mc V$, $k\in\tilde{\mc Z}_L$, 
such that the Zakharov-Shabat equation \eqref{eq:zs} holds for all $h,k\in\tilde{\mc Z}_L$.
Then, the Lax equation \eqref{eq:lax} holds as well, for every $k\in\tilde{\mc Z}_L$. 
In particular, $L$ is a Lax operator and  $\tilde{\mc Z}_L\subset\mc Z_L$.
\item
Suppose $L(\partial)$ is a Lax operator. A collection of functions 
$\varphi_\alpha\in\mc F$, $\alpha\in I$,
is a solution of the hierarchy of Lax equations \eqref{eq:lax} 
if and only if
it is a solution of the Zakharov-Shabat equations \eqref{eq:zs}
for every $h,k\in\mc Z_L$.
Furthermore, if $\varphi_\alpha\in\mc F$, $\alpha\in I$,
is a solution of either \eqref{eq:lax} or \eqref{eq:zs},
then it is a solution of the complementary Zakharov-Shabat equations \eqref{eq:czs}.
\end{enumerate}
\end{proposition}
\begin{proof}
(See \cite{Shi86})
Let $N\geq1$ be the order of $L(\partial)$. If $h,k\in\mc Z_L$, then by \eqref{20170720:eq6} we have
\begin{equation}\label{20210329:eqC}
\frac{\partial B_k(\partial)}{\partial t_h}+\frac{\partial B_k^{(-)}(\partial)}{\partial t_h}=
\frac{\partial L^{\frac kN}(\partial)}{\partial t_h}
=\big[B_h(\partial),L^{\frac kN}(\partial)\big]=\big[B_h(\partial),B_k(\partial)+B_k^{(-)}(\partial)\big]
\,.
\end{equation}
As a consequence, we get
\begin{equation}\label{20170720:eq7}
\begin{array}{l}
\displaystyle{
\vphantom{\Big(}
\Big(
\frac{\partial B_k(\partial)}{\partial t_h}-\frac{\partial B_h(\partial)}{\partial t_k}
+\big[B_k(\partial),B_h(\partial)\big]
\Big)
}\\
\displaystyle{
\vphantom{\Big(}
+
\Big(
\frac{\partial B_k^{(-)}(\partial)}{\partial t_h}-\frac{\partial B_h^{(-)}(\partial)}{\partial t_k}
-\big[B_k^{(-)}(\partial),B_h^{(-)}(\partial)\big]
\Big)
} \\
\displaystyle{
\vphantom{\Big(}
=\big[B_h(\partial),B_k(\partial)+B_k^{(-)}(\partial)\big]
-\big[B_k(\partial),B_h(\partial)+B_h(\partial)^{(-)}\big]
} \\
\displaystyle{
\vphantom{\Big(}
+\big[B_k(\partial),B_h(\partial)\big]
-\big[B_k^{(-)}(\partial),B_h^{(-)}(\partial)\big]
} \\
\displaystyle{
\vphantom{\Big(}
=
-\big[B_k(\partial)+B_k^{(-)}(\partial),B_h(\partial)+B_h^{(-)}(\partial)\big]
=
-[L^{\frac hN}(\partial),L^{\frac kN}(\partial)]
=0
\,.}
\end{array}
\end{equation}
Note that the two summands in the LHS of \eqref{20170720:eq7}
are respectively in $\mc V[\partial]$ and $\mc V[[\partial^{-1}]]\partial^{-1}$.
Hence, they must both vanish, proving part (a).

Next, let us prove part (b). For $k,h\in\tilde{\mc Z}_L$, we have, by \eqref{eq:zs},
\begin{equation}\label{20170720:eq8}
\begin{array}{l}
\displaystyle{
\vphantom{\Big(}
\frac{\partial L^\frac{k}{N}(\partial)}{\partial t_h}-\big[B_h(\partial),L^{\frac kN}(\partial)\big]
=
\frac{\partial B_k(\partial)}{\partial t_h}+\frac{\partial B_k^{(-)}(\partial)}{\partial t_h}
-\big[B_h(\partial),B_k(\partial)+B_k^{(-)}(\partial)\big]
}
\\
\displaystyle{
\vphantom{\Big(}
=
\frac{\partial B_k(\partial)}{\partial t_h}-\big[B_h(\partial),B_k(\partial)\big]+\frac{\partial B_k^{(-)}(\partial)}{\partial t_h}
-\big[B_h(\partial),B_k^{(-)}(\partial)\big]
}
\\
\displaystyle{
\vphantom{\Big(}
=
\frac{\partial B_h(\partial)}{\partial t_k}+\frac{\partial B_k^{(-)}(\partial)}{\partial t_h}
-\big[B_h(\partial),B_k^{(-)}(\partial)\big]
\,.}
\end{array}
\end{equation}
Since $B_k^{(-)}(\partial)$ has negative order,
the RHS of \eqref{20170720:eq8}
is a pseudodifferential operator of order less than or equal to $h-1$
(independently of $k$).
Assume, by contradiction, that
$$
R(\partial)
=
\frac{\partial L^{\frac1N}(\partial)}{\partial t_h}
-
\big[B_h(\partial),L^{\frac1N}(\partial)\big]
$$
is a pseudodifferential operator of order $d\in\mb Z$
and leading coefficient $v\in\mc V\backslash\{0\}$.
By the Leibniz rule, we have
\begin{equation}\label{20170720:eq9}
\frac{\partial L^\frac{k}{N}(\partial)}{\partial t_h}-\big[B_h(\partial),L^{\frac kN}(\partial)\big]
=
\sum_{i=0}^{k-1}
L^\frac{i}{N}(\partial)R(\partial)L^{\frac{k-1-i}{N}}(\partial)
\,.
\end{equation}
The RHS of \eqref{20170720:eq9} has order equal to $d+k-1$,
the leading coefficient being equal to $kv\neq0$.
Since the LHS of \eqref{20170720:eq8} and \eqref{20170720:eq9} coincide,
we deduce that $d+(k-1)N\leq h-1$ for every $k\in\tilde{\mc Z}_L$, which is impossible since $\tilde{\mc Z}_L$ is infinite.
Hence, $R$ must be $0$. It follows, setting $k=N$ in \eqref{20170720:eq9}, 
that the Lax equation \eqref{eq:lax} holds for every $k\in\tilde{\mc Z}_L$,
so that $\tilde{\mc Z}_L\subset\mc Z_L$, as claimed.

The proof of part (c) is the same as for parts (a) and (b),
by replacing everywhere $L(\partial)$, $B_k(\partial)$ and $B_h(\partial)$
with their evaluations $L(\varphi;\partial)$, $B_k(\varphi;\partial)$
and $B_h(\varphi;\partial)$ in $\mc F((\partial^{-1}))$.
\end{proof}

\begin{example}\label{exa:KP}
Let $L(\partial)=\partial+\sum_{i=1}^{\infty}u_i\partial^{-i}$ be the Sato Lax operator.
By a straightforward computation we have
$B_1(\partial)=\partial$, $B_2(\partial)=\partial^2+2u_1$ and $B_3(\partial)=\partial^3+3u_1\partial+3u_2+3u_1'$.
The Zakharov-Shabat equation \eqref{eq:zs} for $k=2$ and $h=3$ gives
\begin{equation}\label{20210121:eq1}
\left(6u_2'+3u_1''-3\frac{\partial u_1}{\partial t_2}\right)\partial
+2\frac{\partial u_1}{\partial t_3}-3\frac{\partial u_2}{\partial t_2}-3\frac{\partial u_1'}{\partial t_2}-6u_1u_1'+3u_2''+u_1'''
\,.
\end{equation}
This is equivalent to
$$
\frac{\partial u_1}{\partial t_2}=2u_2'+u_1''\,,\qquad
2\frac{\partial u_1}{\partial t_3}-3\frac{\partial u_2}{\partial t_2}=6u_1u_1'+3u_2''+2u_1'''=0\,.
$$
Applying $\frac{\partial}{\partial t_2}$ to the first equation and $\partial$ to the second we get
$$
\frac{\partial^2 u_1}{\partial t_2^2}=2\frac{\partial u_2'}{\partial t_2}+\frac{\partial u_1''}{\partial t_2}
=2\frac{\partial u_2'}{\partial t_2} +2u_2'''+u_1^{(4)}
$$
and
$$
2\frac{\partial u_1'}{\partial t_3}=3\frac{\partial u_2'}{\partial t_2}+6(u_1\partial u_1)'+3u_2'''+2 u_1^{(4)}
\,,
$$
from which follows that
$$
3\frac{\partial^2 u_1}{\partial t_2^2}
-4\partial \frac{\partial u_1}{\partial t_3}
=-\partial^2\left(6u_1^2+\partial^2u_1\right)
\,.
$$
Renaming $t_2=y$, $t_3=t$, and $u=2u_1$, we get
$$
3 \frac{\partial^2 u}{\partial y^2}
=
\left(4 \frac{\partial u}{\partial t} - u''' - 6u u'\right)'
\,.
$$
which is the Kadomtsev-Petviashvili (KP) equation.
\end{example}

\label{sec:5}
\section{Lax and Sato equations, linear problem and bilinear equation}

Let $L(\partial)\in\mc V((\partial^{-1}))$ be as in \eqref{eq:L}. Let, as in Section \ref{sec:2.2}, $\mc F$ be an algebra of functions
in space-time, with time derivatives $\frac{\partial}{\partial t_k}$ indexed by $k\in\mc Z_L$.

\subsection{Lax equations and Sato equations}
\label{sec:3.4}

\begin{theorem}\label{thm:sato}
Let $\varphi=(\varphi_\alpha)_{\alpha\in I}$ be a collection of functions 
$\varphi_\alpha\in\mc F$.
\begin{enumerate}[(a)]
\item
If there exists a monic pseudodifferential operator of order $0$ over the algebra of functions on space-time $\mc F$,
\begin{equation}\label{20210407:S}
S(\partial)=1+\sum_{i=1}^{\infty}s_i\partial^{-i}\in\mc F((\partial^{-1}))\,,
\end{equation}
satisfying the following 
\emph{dressing equation}
\begin{equation}\label{eq:dress}
L(\varphi;\partial)=S(\partial)\partial^N S^{-1}(\partial)
\,,
\end{equation}
and \emph{Sato equations}
\begin{equation}\label{eq:sato}
\frac{\partial S(\partial)}{\partial t_k}
=
-B_k^{(-)}(\varphi;\partial)S(\partial)
\,,
\quad\text{for every } k\in\mc Z_L\,,
\end{equation}
then $\varphi$ is a solution of the hierarchy of Lax equations \eqref{eq:lax}.
\item
Conversely, assume that the algebra of functions $\mc F$ is integrable (cf. Definition \ref{def:integrable}) 
and assume that $\varphi$ is a solution of the hierarchy of Lax equations \eqref{eq:lax} 
such that $a_0(\varphi)=0$
($a_0(\varphi)\in\mc F$ being the evaluation at $u=\varphi$).
Then, 
for every monic constant coefficients pseudodifferential operator $S_0(\partial)\in\mb F((\partial^{-1}))$ of order $0$,
there exists a unique order 0 monic pseudodifferential operator $S(\partial)\in\mc F((\partial^{-1}))$ as in \eqref{20210407:S} 
satisfying the dressing equation \eqref{eq:dress} 
and the Sato equations \eqref{eq:sato},
and such that $S(\partial)|_{x=t=0}=S_0(\partial)$.
\end{enumerate}
\end{theorem}
\begin{proof}
(See \cite{Shi86})
Applying $\frac{\partial}{\partial t_k}:\,\mc F\to\mc F$
to both sides of the dressing equation \eqref{eq:dress},
and using the Sato equations \eqref{eq:sato}, we get
$$
\begin{array}{l}
\displaystyle{
\vphantom{\Big(}
\frac{\partial L(\varphi;\partial)}{\partial t_k}
=
\frac{\partial S(\partial)}{\partial t_k}\partial^N S^{-1}(\partial)
-
S(\partial)\partial^N S^{-1}(\partial)\frac{\partial S(\partial)}{\partial t_k}S^{-1}(\partial)
} \\
\displaystyle{
\vphantom{\Big(}
=
-B_k^{(-)}(\varphi;\partial)S(\partial)\partial^N S^{-1}(\partial)
+
S(\partial)\partial^N S^{-1}(\partial)B_k^{(-)}(\varphi;\partial)
} \\
\displaystyle{
\vphantom{\Big(}
=
-[B_k^{(-)}(\varphi;\partial),L(\varphi;\partial)]
=[B_k(\varphi;\partial),L(\varphi;\partial)]
\,.}
\end{array}
$$
This proves (a).

In order to prove part (b), assume that $\varphi_\alpha\in\mc F$, $\alpha\in I$, is a solution 
of the Lax equations \eqref{eq:lax}.
The dressing equation \eqref{eq:dress}
can be rewritten as
\begin{equation}\label{20170725:eq1}
\begin{array}{l}
\displaystyle{
\vphantom{\Big(}
(1+s_1\partial^{-1}+s_2\partial^{-2}+\dots)\partial^N
} \\
\displaystyle{
\vphantom{\Big(}
=
(\partial^N+a_0(\varphi)\partial^{N-1}+a_1(\varphi)\partial^{N-2}+\dots)
(1+s_1\partial^{-1}+s_2\partial^{-2}+\dots)
\,.}
\end{array}
\end{equation}
By looking at the coefficient of $\partial^{N-1}$ in both sides of equation \eqref{20170725:eq1},
we get $a_0(\varphi)=0$,
while, by looking at the coefficient of $\partial^{N-1-i}$,
for $i\geq1$, in both sides of \eqref{20170725:eq1},
we get
\begin{equation}\label{20170725:eq2}
s_i'
=
-\frac1N a_i(\varphi)
-\frac1N\sum_{n=0}^{i-2}\binom{N}{n+2}s_{i-n-1}^{(n)}
-\frac1N\sum_{j=1}^{i-1}\sum_{n=0}^{i-j-1}\binom{N-1-j}{n}a_j(\varphi)s_{i-j-n}^{(n)}
\,.
\end{equation}
Next, the Sato equations \eqref{eq:sato}
can be rewritten as
\begin{equation}\label{20170725:eq3}
\begin{array}{l}
\displaystyle{
\vphantom{\Big(}
\frac{\partial s_1}{\partial t_k}\partial^{-1}+\frac{\partial s_2}{\partial t_k}\partial^{-2}+\dots
} \\
\displaystyle{
\vphantom{\Big(}
=
(c_{1;k}(\varphi)\partial^{-1}+c_{2;k}(\varphi)\partial^{-2}+\dots)
(1+s_1\partial^{-1}+s_2\partial^{-2}+\dots)
\,,}
\end{array}
\end{equation}
where we are letting $-B_k^{(-)}(\partial)=\sum_{j=1}^\infty c_{j;k}\partial^{-j}\in\mc V[[\partial^{-1}]]\partial^{-1}$,
and we denote, as usual, by $c_{j;k}(\varphi)\in\mc F$ the evaluation of $c_{j;k}\in\mc V$ at $u=\varphi$.
By looking at the coefficient of $\partial^{-i}$, for $i\geq1$, 
in both sides of equation \eqref{20170725:eq3},
we get
\begin{equation}\label{20170725:eq4}
\frac{\partial s_i}{\partial t_k}
=
c_{i;k}(\varphi)
+
\sum_{j=1}^{i-1}\sum_{n=0}^{i-j-1}(-1)^n\binom{j+n-1}{n}
c_{j;k}(\varphi)s_{i-j-n}^{(n)}
\,\,,\,\,\,\,
k\in\mc Z_L
\,.
\end{equation}
For every $i\geq1$, equations \eqref{20170725:eq2} and \eqref{20170725:eq4}
form a system of equations in the only unknown function $s_i\in\mc F$,
if we assume that we have solved, recursively, the previous equations on $s_j$ with $j<i$.
Indeed, the RHS of both \eqref{20170725:eq2} and \eqref{20170725:eq4}
only involves the functions $a_j(\varphi)\in\mc F$, which are given,
and functions $s_j^{(n)}$ with $j<i$, which are assumed to be known by the recursive construction.
By the integrability assumption on $\mc F$ (cf. Definition \ref{def:integrable}),
this system has a unique solution $s_i\in\mc F$ satisfying the initial condition $s_i(0)=s_{i;0}$,
the coefficient of $\partial^{-i}$ in $S_0(\partial)$,
provided that the compatibility conditions \eqref{20170721:eq2b} hold.
Recalling that \eqref{20170725:eq2} is equivalent to the dressing equation 
$S(\partial)\partial^NS^{-1}(\partial)=L(\varphi;\partial)$,
the first of the compatibility conditions \eqref{20170721:eq2b}
for the system \eqref{20170725:eq2}--\eqref{20170725:eq4}
amounts to
$$
\frac{\partial L(\varphi;\partial)}{\partial t_k}
=
-[B_k^{(-)}(\varphi;\partial),L(\varphi;\partial)]
\,,
$$
which holds since, by assumption, $\varphi$ is a solution of the Lax equations \eqref{eq:lax}.
On the other hand, \eqref{20170725:eq4} is equivalent to the Sato equation 
$\frac{\partial S(\partial)}{\partial t_k}=-B_k^{(-)}(\varphi;\partial)S(\partial)$,
hence, the second of the compatibility conditions \eqref{20170721:eq2b}
for the system \eqref{20170725:eq2}--\eqref{20170725:eq4}
amounts to
$$
\begin{array}{l}
\displaystyle{
\vphantom{\Big(}
-\frac{\partial B_k^{(-)}(\varphi;\partial)}{\partial t_h}S(\partial)
+B_k^{(-)}(\varphi;\partial)B_h^{(-)}(\varphi;\partial)S(\partial)
}
\\
\displaystyle{
\vphantom{\Big(}
=
-\frac{\partial B_h^{(-)}(\varphi;\partial)}{\partial t_k}S(\partial)
+B_h^{(-)}(\varphi;\partial)B_k^{(-)}(\varphi;\partial)S(\partial)
\,,
}
\end{array}
$$
which is the complementary Zakharov-Shabat equation \eqref{eq:czs},
and therefore holds by Proposition \ref{20170720:prop}(c).
\end{proof}
\begin{remark}\label{20170725:rem}
Note that, by the uniqueness of the $N$-th root of a monic pseudodifferential operator
(cf. Lemma \ref{20170720:lem}(b)),
the dressing equation 
$L(\varphi;\partial)=S(\partial)\partial^N S^{-1}(\partial)$
on the pseudodifferential operator $S(\partial)\in\mc F((\partial^{-1}))$, monic of order $0$,
is equivalent to the equation
\begin{equation}\label{eq:satoN}
L^{\frac1N}(\varphi;\partial)=S(\partial)\partial S^{-1}(\partial)
\,.
\end{equation}
\end{remark}

\subsection{Sato equation and the linear problem for the wave function}

Consider the space of oscillating functions $\mc F((z^{-1})) e^{z\cdot\bm t}$ 
defined in Section \ref{sec:2.6}.

\begin{theorem}\label{thm:linear}
Let $\varphi=(\varphi_\alpha)_{\alpha\in I}$, $\varphi_\alpha\in\mc F$.
Let $S(\partial)\in\mc F((\partial^{-1}))$ be a monic pseudodifferential operator of order $0$,
and let $w(z)=S(z)e^{z\cdot\bm t}\in\mc F((z^{-1}))e^{z\cdot\bm t}$ be the corresponding oscillating function.
\begin{enumerate}[(a)]
\item
The dressing equation \eqref{eq:dress} on $S(\partial)$
is equivalent to the following \emph{eigenvalue problem} on the \emph{wave function} $w(z)$:
\begin{equation}\label{eq:eigen}
L(\varphi;\partial)w(z)=z^Nw(z)
\,.
\end{equation}
\item
Assuming that the dressing equation \eqref{eq:dress} holds,
the Sato equation \eqref{eq:sato} for $S(\partial)$ is equivalent to the following
\emph{linear problem} for the wave function $w(z)$:
\begin{equation}\label{eq:linear}
\frac{\partial w(z)}{\partial t_k}
=
B_k(\varphi;\partial)w(z)
\,,
\quad
k\in\mc Z_L
\,.
\end{equation}
\end{enumerate}
\end{theorem}
\begin{proof}
Equation \eqref{eq:eigen} can be rewritten, using the first equation in \eqref{20210407:eq1}, as
\begin{equation}\label{20170726:eq4}
\big(L(\varphi;\partial)S(\partial)-S(\partial)\partial^N\big)e^{z\cdot\bm t}
=0
\,.
\end{equation}
Claim (a) follows by Lemma \ref{20170725:lem}.
Next, let us prove claim (b).
Note that, by Remark \ref{20170725:rem},
the dressing equation \eqref{eq:dress}
is also equivalent to 
$$
L^{\frac1N}(\varphi;\partial)w(z)=zw(z)
\,.
$$
Hence we have
\begin{equation}\label{20170725:eq7}
z^kw(z)=L^{\frac kN}(\varphi;\partial)w(z)
=\left(B_k(\varphi;\partial)+B_k^{(-)}(\varphi;\partial)\right)w(z)\,.
\end{equation}
By \eqref{20170725:eq5} and \eqref{20170725:eq7},
we get
$$
\begin{array}{l}
\displaystyle{
\vphantom{\Big(}
\frac{\partial w(z)}{\partial t_k}=\Big(\frac{\partial S(\partial)}{\partial t_k}+z^kS(z)\Big)e^{z\cdot\bm t}
=
\Big(\frac{\partial S(\partial)}{\partial t_k} S^{-1}(\partial)+z^k\Big)w(z)
} \\
\displaystyle{
\vphantom{\Big(}
=
\Big(\frac{\partial S(\partial)}{\partial t_k} S^{-1}(\partial)+B_k(\varphi;\partial)+B_k^{(-)}(\varphi;\partial)\Big)
w(z)
\,.}
\end{array}
$$
Hence, if Sato equation \eqref{eq:sato} holds, we automatically get the linear problem \eqref{eq:linear}.
Conversely, if the linear problem \eqref{eq:linear} holds,
Sato equation \eqref{eq:sato} holds due to Lemma \ref{20170725:lem}.
\end{proof}

\subsection{Lax equations and the bilinear equation on the wave function}
\label{sec:3.6}

\begin{theorem}\label{thm:bilinear}
Let $S(\partial)\in\mc F((\partial^{-1}))$ be an invertible pseudodifferential operator,
let $w(z)=S(z)e^{z\cdot\bm t}\in\mc F((z^{-1}))e^{z\cdot\bm t}$ be the corresponding oscillating function,
and let $w^\star(z)=(S^*)^{-1}(-z)e^{-z\cdot\bm t}$ be the corresponding adjoint anti-oscillating function
(cf. \eqref{20170725:eq8}).
Then, the following conditions are equivalent:
\begin{enumerate}[(i)]
\item
For every $k\in\mc Z_L$, we have
$$
\Big(
\frac{\partial S(\partial)}{\partial t_k}S^{-1}(\partial)+S(\partial)\partial^k S^{-1}(\partial)
\Big)_-=0
\,.
$$
\item
For every $k\in\mc Z_L$, there exists a differential operator $\tilde B_k(\partial)\in\mc F[\partial]$ such that
$$
\frac{\partial w(z)}{\partial t_k}=\tilde B_k(\partial)w(z)
\,.
$$
\item
For every $s\geq0$, $k_1,\dots,k_s\in\mc Z_L$ and $n_0,n_1,\dots,n_s\geq0$, we have
$$
\Res_z\Big(
\frac{\partial^{n_0+n_1+\dots+n_s}w(z)}{\partial x^{n_0}\partial t_{k_1}^{n_1}\cdots\partial t_{k_s}^{n_s}}w^\star(z)
\Big)=0
\,.
$$
\end{enumerate}
\end{theorem}
\begin{remark}\label{20170726:rem2}
Condition (iii) can be suggestively rewritten as the following \emph{bilinear equation}
on the wave function $w(z)=w(x,\bm t;z)$:
\begin{equation}\label{eq:bilinear}
\Res_z\big(
w(x,\bm t;z)w^\star(x^\prime,\bm t^\prime;z)
\big)
=0
\,.
\end{equation}
Indeed, the equation in (iii) formally coincides with the coefficient of
$$
\frac{(x-x')^{n_0}(t_{k_1}-t_{k_1}')^{n_1}\dots(t_{k_s}-t_{k_s}')^{n_s}}{n_0!n_1!\dots n_s!}
$$
in the Taylor series expansion
of \eqref{eq:bilinear} around $x=x'$, $\bm t=\bm t'$, see \cite{Dic03}.
\end{remark}
\begin{proof}[Proof of Theorem \ref{thm:bilinear}]
By the definition \eqref{20170725:eq5} of the action of $\frac{\partial}{\partial t_k}$
on oscillating functions, we have
$$
\frac{\partial w(z)}{\partial t_k}
=
\Big(\frac{\partial S(\partial)}{\partial t_k}S^{-1}(\partial)+S(\partial)\partial^kS^{-1}(\partial)\Big)
w(z)\,.
$$
The equivalence of conditions (i) and (ii) is an obvious consequence of this identity
and of Lemma \ref{20170725:lem}.
Indeed, when this holds, we have
$$
\tilde B_k(\partial)=\Big(\frac{\partial S(\partial)}{\partial t_k}S^{-1}(\partial)+S(\partial)\partial^kS^{-1}(\partial)\Big)_+
\,.
$$

Next, let us prove that condition (ii) implies (iii).
We claim, by induction on $n_0+n_1+\dots+n_s$, that
\begin{equation}\label{20170726:eq3}
\frac{\partial^{n_0+n_1+\dots+n_s}w(z)}{\partial x^{n_0}\partial t_{k_1}^{n_1}\cdots\partial t_{k_s}^{n_s}}
=
P(\partial)w(z)
\end{equation}
for some differential operator $P(\partial)\in\mc F[\partial]$
(depending on $n_0,n_1,\dots,n_s,k_1,\dots,k_s$).
For $n_0=1$, $n_1=\dots=n_s=0$, \eqref{20170726:eq3} holds with $P(\partial)=\partial$. For $n_0=0$ and
$n_1+\dots+n_s=1$, \eqref{20170726:eq3} follows by condition (ii).
Assuming, by induction, that \eqref{20170726:eq3} holds,
we have
$$
\frac{\partial}{\partial x}
\frac{\partial^{n_0+n_1+\dots+n_s}w(z)}{\partial x^{n_0}\partial t_{k_1}^{n_1}\cdots\partial t_{k_s}^{n_s}}
=\partial P(\partial) w(z)
\,,
$$
and, for $h\in\mc Z_L$,
$$
\frac{\partial}{\partial t_h}
\frac{\partial^{n_0+n_1+\dots+n_k}w(z)}{\partial x^{n_0}\partial t_{k_1}^{n_1}\cdots\partial t_{k_s}^{n_s}}
=
\Big(
\frac{\partial P(\partial)}{\partial t_h}
+P(\partial)\tilde B_h(\partial)
\Big)w(z)
\,,
$$
proving \eqref{20170726:eq3}.
Then we have
$$
\begin{array}{l}
\displaystyle{
\vphantom{\Big(}
\Res_z\Big(
\frac{\partial^{n_0+n_1+\dots+n_s}w(z)}{\partial x^{n_0}\partial t_{k_1}^{n_1}\cdots\partial t_{k_s}^{n_s}}w^\star(z)
\Big)
=
\Res_z
\big(P(\partial)w(z)\big)w^\star(z)
} \\
\displaystyle{
\vphantom{\Big(}
=
\Res_z
(P\circ S)(z)(S^*)^{-1}(-z)
=
\Res_z
P(z)
=
0
\,.}
\end{array}
$$
In the third equality we used Lemma \ref{lem:hn1a}.

Finally, we prove that condition (iii) implies (i).
For every $n\geq0$ and $k\in\mc Z_L$, we have,
by condition (iii) and the definition \eqref{20170725:eq5} 
of the action of $\frac{\partial}{\partial t_k}$ on oscillating functions,
$$
0
=
\Res_z
\big(\partial^n\frac{\partial w(z)}{\partial t_k}\big)w^\star(z)
=
\Res_z
\Big(
(z+\partial)^n
\Big(\frac{\partial S(z)}{\partial t_k}+S(z)z^k\Big)\Big)(S^*)^{-1}(-z)
\,.
$$
We can now use Lemma \ref{20170726:lem}
to conclude that condition (i) holds.
\end{proof}

\subsection{Summarizing statement}
\label{sec:3.7}

We can summarize the results of Sections \ref{sec:3.4}--\ref{sec:3.6}
in the following Corollary:
\begin{corollary}\label{cor:equiv}
Let $\varphi=(\varphi_\alpha)_{\alpha\in I}\in\mc F^\ell$.
Let $S(\partial)\in\mc F((\partial^{-1}))$ be a monic pseudodifferential operator of order $0$
and let $w(z)=S(z)e^{z\cdot\bm t}\in\mc F((z^{-1}))e^{z\cdot\bm t}$ be the corresponding oscillating function.
Then, the dressing equation \eqref{eq:dress} on $S(\partial)$ is equivalent 
to the eigenvalue problem \eqref{eq:eigen} on the wave function $w(z)$.
Furthermore, if either \eqref{eq:dress} or \eqref{eq:eigen} holds,
then the following conditions are equivalent:
\begin{enumerate}[(a)]
\item
the Sato equation \eqref{eq:sato} holds;
\item
the linear problem \eqref{eq:linear} holds;
\item
the bilinear equation \eqref{eq:bilinear} holds.
\end{enumerate}
If, moreover, any of the three equivalent conditions (a)--(c) hold,
then $\varphi$ is a solution to the Lax equation \eqref{eq:lax}
(i.e. \eqref{eq:lax-sol} holds).
\end{corollary}
\begin{proof}
The equivalence of \eqref{eq:dress} and \eqref{eq:eigen} is given by Theorem \ref{thm:linear}(a).
Moreover, assuming \eqref{eq:dress},
the equivalence of (a) and (b) is given by Theorem \ref{thm:linear}(b),
and we need to prove their equivalence to (c).
Theorem \ref{thm:bilinear} says, in particular, that (a) implies (c).
Let us prove the converse implication.
Since $S(\partial)$ is monic of order $0$,
then $\frac{\partial S(\partial)}{\partial t_k}S^{-1}(\partial)$ 
is a pseudodifferential operator of negative order.
Hence, by Theorem \ref{thm:bilinear}(i) and \eqref{eq:satoN}, we have
$$
\frac{\partial S(\partial)}{\partial t_k}S^{-1}(\partial)
=
-\big(S(\partial)\partial^kS^{-1}(\partial)\big)_-
=-\big(L^{\frac kN}(\varphi,\partial)\big)_-
=-B_k^{(-)}(\varphi;\partial)
\,,
$$
proving the Sato equation \eqref{eq:sato}.
Finally, the last assertion of the Theorem is given by Theorem \ref{thm:sato}(a).
\end{proof}

\section{Wave functions and tau-functions}
\label{sec:3.8}

\subsection{Tau-functions of KP type}

We review in this section the construction of tau-functions when $\mc Z_L=\mb Z_{\geq1}\,(=\mc Z)$, see \cite{DJKM83,Dic03}.

Throughout this section we assume that 
the algebra of functions $\mc F$ is integrable (cf. Definition \ref{def:integrable}).
Consider the operator
$$
e^{-z^{-1}\cdot\tilde{\partial}_{\bm t}}:\mc F\to\mc F((z^{-1}))
\,,
\text{ where }
\,
z^{-1}\cdot\tilde{\partial}_{\bm t}\,:=\sum_{k=1}^{\infty}\frac{z^{-k}}{k}\frac{\partial}{\partial t_k}
\,.
$$
By Taylor expansion, it ``shifts'' the time variables $t_k\mapsto t_k-\frac1{kz^k}$.
Fix a monic pseudodifferential operator of order $0$, $S(\partial)\in1+\mc F[[\partial^{-1}]]\partial^{-1}$. 
Let $w(z)=S(z)e^{z\cdot{\bm t}}\in\mc F((z^{-1}))e^{z\cdot{\bm t}}$ be the corresponding oscillating function,
and let $w^\star(z)=(S^*)^{-1}(-z)e^{-z\cdot{\bm t}}$ be the corresponding adjoint anti-oscillating function.
\begin{theorem}\label{20170912:cor1}
If the bilinear identity \eqref{eq:bilinear} holds,
then, there exists a function $\tau\in\mc F$ such that
\begin{equation}\label{20170912:eq1}
w(z)=\frac{\Gamma({\bm t};z)\tau}{\tau}\qquad\text{and}\qquad
w^\star(z)=\frac{\Gamma^\star({\bm t};z)\tau}{\tau}
\,,
\end{equation}
where $\Gamma({\bm t};z)=e^{z\cdot{\bm t}}e^{-z^{-1}\cdot\tilde\partial_{\bm t}}:\mc F\to\mc F((z^{-1}))e^{z\cdot{\bm t}}$
and $\Gamma^{\star}({\bm t};z)=e^{-z\cdot{\bm t}}e^{z^{-1}\cdot\tilde\partial_{\bm t}}:\mc F\to\mc F((z^{-1}))e^{-z\cdot{\bm t}}$
are called vertex operators.
Moreover, the function $\tau\in\mc F$ solving \eqref{20170912:eq1} is unique up to multiplication by a function of $x$
(constant in $\bm t$).
\end{theorem}

As an immediate corollary of Theorem \ref{20170912:cor1} and Corollary \ref{cor:equiv},
we get the following
\begin{corollary}\label{20170912:cor2}
Let $L(\partial)\in\mc V((\partial^{-1}))$ be as in \eqref{eq:L}
and let $\varphi=(\varphi_\alpha)_{\alpha\in I}$, $\varphi_\alpha\in\mc F$,
be a solution of the hierarchy of Lax equations \eqref{eq:lax-sol} such that $a_0(\varphi)=0$.
Then, there exists $\tau\in\mc F$
such that
\begin{equation}\label{eq:cor2}
L(\varphi;\partial)=S(\partial)\partial^N S(\partial)^{-1}
\,\,\text{ with }\,\,
S(z)=\frac{e^{-z^{-1}\cdot\tilde\partial_{\bm t}}\tau}{\tau}
\,.
\end{equation}
\end{corollary}
\begin{proof}
By Theorem \ref{thm:sato}(b), there exists $S(\partial)\in1+\mc F[[\partial^{-1}]]\partial^{-1}$
satisfying the dressing equation \eqref{eq:dress},
and by Theorem \ref{20170912:cor1} there exists $\tau\in\mc F$
such that \eqref{eq:cor2} holds.
\end{proof}

Corollary \ref{20170912:cor2} is saying that the Lax equations,
which are systems of equations in (infinitely) many unknown functions $a_i$,
can be ``reduced'' to a problem involving one single unknown function $\tau$.
In fact, we can translate the bilinear equation \eqref{eq:bilinear}
(which, by Corollary \ref{cor:equiv}, is essentially equivalent to the Lax equations \eqref{eq:lax}),
to a system of partial differential equations on the unknown function $\tau$.
The bilinear identity \eqref{eq:bilinear} is translated to the following bilinear identity for $\tau$:
$$
\res_z \left(\Gamma({\bm t}';z)\tau\right)\left(\Gamma^\star({\bm t}'';z)\tau\right)
=\res_z(e^{-z^{-1}\cdot\tilde\partial_{{\bm t}'}}\tau)(e^{z\cdot\tilde\partial_{{\bm t}''}}\tau)e^{z\cdot({\bm t}'-{\bm t}'')}=0
\,.
$$
By computing the above equation for ${\bm t}'={\bm t}-{\bm y}$ and ${\bm t}''={\bm t}+{\bm y}$, we get
\begin{equation}\label{20170920:eq1}
\sum_{k=0}^{\infty}p_k(-2{\bm y})p_{k+1}(\tilde \partial_{\bm y})(e^{-{\bm y}\cdot\partial_{{\bm t}}}\tau)(e^{{\bm y}\cdot\partial_{{\bm t}}}\tau)
=0
\,,
\end{equation}
where the Schur polynomials $p_k({\bm y})\in\mb F[y_1,\dots,y_k]$ are defined by
\begin{equation}\label{eq:schur}
e^{\sum_{k=1}^{\infty}y_k z^k}=\sum_{k=0}^\infty p_k({\bm y})z^k
\,.
\end{equation}
Equation \eqref{20170920:eq1} indeed provides an infinite system of partial differential equations on $\tau$,
by looking at the coefficient of each monomial $y_1^{n_1}y_2^{n_2}\dots y_s^{n_s}$.

We can express the solution $L(\varphi;\partial)$ in terms of the function $\tau$.
Recalling that $w(z)=S(z)e^{z\cdot t}$, using equations \eqref{20170912:eq1} and \eqref{eq:schur} we get
$S(z)=\sum_{k=0}^{\infty}\frac{p_k(-\tilde\partial_t)\tau}{\tau}z^{-k}$.
Similarly, by $w^\star(z)=(S^*)^{-1}(-z)e^{-z\cdot t}$ we get
$(S^*)^{-1}(-z)=\sum_{k=0}^{\infty}\frac{p_k(\tilde\partial_t)\tau}{\tau}z^{-k}$,
from which follows that
$S^{-1}(z)=\sum_{k=0}^{\infty}(z+\partial)^{-k}\frac{p_k(\tilde\partial_t)\tau}{\tau}$.
Hence, if $\tau$ is a tau-function of the solution $\varphi$ of the Lax equation \eqref{eq:lax} such that $a_0(\varphi)=0$, 
we obtain
\begin{equation}
\begin{split}\label{eq:befana}
L(\varphi;\partial)&=\sum_{h,k=0}^{\infty}\frac{p_h(-\tilde\partial_t)\tau}{\tau}\partial^{N-h-k}\circ\frac{p_k(\tilde\partial_t)\tau}{\tau}
\\
&=
\sum_{k=0}^{\infty}\left(
\sum_{h=0}^k\sum_{l=0}^{k-h}\binom{N-k+l}{l}\frac{p_{k-h-l}(-\tilde\partial_t)\tau}{\tau}
\partial^l\left(\frac{p_h(\tilde\partial_t)\tau}{\tau}\right)
\right)\partial^{N-k}\,.
\end{split}
\end{equation}
By equating powers of $\partial^{N-2}$ in both sides of the above formula we get for example
$$
a_1(\varphi)=N\partial^2\log\tau
\,.
$$
(We used the fact that $\partial\log\tau=-s_1=\frac{\partial}{\partial t_1}\log\tau$, cf. equation \eqref{eq:partial-tau} below.)

\medskip

The remainder of the present section is devoted to the proof of Theorem \ref{20170912:cor1}.
Equation \eqref{20170912:eq1} defining $\tau$,
which we want to solve,
is obviously equivalent to the  the equation relating $S(z)$ and $\tau$ in \eqref{eq:cor2}.
Applying the logarithm to both sides of the latter, we get
\begin{equation}\label{eq:def-tau1}
\log(S(z))=\big(e^{-z^{-1}\cdot\tilde\partial_{\bm t}}-1\big)\log\tau
\,.
\end{equation}
For this, we used the (obvious) fact that the operator $e^{-z^{-1}\cdot\tilde\partial_{\bm t}}$, 
describing the ``shift'' of variables $t_k\mapsto t_k-\frac1{kz^k}$,
commutes with taking the logarithm.
Note that the logarithm of $S(z)$ is well defined by power series expansion in $z^{-1}$,
since $S(z)\in 1+\mc F[[z^{-1}]]z^{-1}$.
The outline of the proof of Theorem \ref{20170912:cor1} is as follows.
From equation \eqref{eq:def-tau1} we shall derive an (equivalent) system of partial differential equations
of type \eqref{20170721:eq2} for $\phi=\log\tau$,
depending on a function $f$ of $x$.
We will then prove that such system is compatible. 
Then, by the integrability assumption on $\mc F$, it will uniquely define $\phi$ up to the choice of the additive function $f$,
i.e. it will uniquely define $\tau$ up to a factor depending only on $x$,
as claimed by Theorem \ref{20170912:cor1}.

\medskip

First, we derive from \eqref{eq:def-tau1} an equation for $\partial\log\tau$.
For this, we need the following result.
\begin{lemma}\label{lem1}
The following identities hold:
\begin{enumerate}[a)]
\item
$S(z)e^{-z^{-1}\cdot\tilde{\partial}_{\bm t}}(S^*)^{-1}(-z)=1$.
\item
$\partial\log S(z)=(1-e^{-z^{-1}\cdot\tilde{\partial}_{\bm t}})s_1$.
\end{enumerate}
\end{lemma}
\begin{proof}
We set $t_k'=t_k-\frac{1}{kw^k}$, $k\geq1$, 
in the bilinear identity \eqref{eq:bilinear}, to get
\begin{equation}\label{20170910:eq4a}
\Res_z\big(
S(x,\bm t;z)
e^{z\cdot\bm t}
e^{-z\cdot\bm t'}
(S^*)^{-1}(x',\bm t';-z)
\big)
=0
\,.
\end{equation}
A straightforward computation leads to
\begin{equation}\label{eq:exp}
e^{z\cdot\bm t}
e^{-z\cdot{\bm t'}}=\iota_w\big(1-\frac zw\big)^{-1}e^{z(x-x')}
\,.
\end{equation}
Hence, equation \eqref{20170910:eq4a} gives
$$
\Res_z\big(
\iota_w(z-w)^{-1}
e^{z(x-x')}
S(x,\bm t;z)
e^{-w^{-1}\cdot\tilde\partial_{\bm t}}
(S^*)^{-1}(x',\bm t;-z)
\big)
=0
\,,
$$
which, by the second equation in \eqref{eq:positive}, gives
\begin{equation}\label{20170910:eq4b}
\Big(
e^{z(x-x')}
S(x,\bm t;z)
e^{-w^{-1}\cdot\tilde\partial_{\bm t}}
(S^*)^{-1}(x',\bm t;-z)
\Big)_-\Big|_{z=w}
=0
\,,
\end{equation}
where the $-$ means that we need to take the negative powers of $z$, before setting $z=w$.
Clearly, we can rewrite \eqref{20170910:eq4b} as
\begin{equation}\label{20170910:eq4c}
\begin{split}
& e^{w(x-x')}
S(x,\bm t;w)
e^{-w^{-1}\cdot\tilde\partial_{\bm t}}
(S^*)^{-1}(x',\bm t;-w) \\
& =
\Big(
e^{z(x-x')}
S(x,\bm t;z)
e^{-w^{-1}\cdot\tilde\partial_{\bm t}}
(S^*)^{-1}(x',\bm t;-z)
\Big)_+\Big|_{z=w}
\,.
\end{split}
\end{equation}
Setting $x'=x$ and recalling that $S(z)\in 1+\mc F[[z^{-1}]]z^{-1}$,
we immediately get claim (a).
Instead, applying $\frac{\partial}{\partial x}$ to both sides of \eqref{20170910:eq4c}
and then setting $x'=x$, the LHS is
$$
\big((w+\partial)S(w)\big)
e^{-w^{-1}\cdot\tilde\partial_{\bm t}}
(S^*)^{-1}(-w) 
=
w+
\frac{\partial S(w)}{S(w)}
=
w+\partial\log S(w)
\,,
$$
where we used part (a).
On the other hand, the RHS of \eqref{20170910:eq4c} gives, since $S(z)=1+s_1z^{-1}+\dots$ and $(S^*)^{-1}(-z)=1-s_1z^{-1}+\dots$,
$$
\Big(
\big((z+\partial)S(z)\big)
e^{-w^{-1}\cdot\tilde\partial_{\bm t}}
(S^*)^{-1}(-z)
\Big)_+\Big|_{z=w}
=
w
+
s_1
-
e^{-w^{-1}\cdot\tilde\partial_{\bm t}} s_1
\,.
$$
Claim (b) follows.
\end{proof}
Comparing claim (b) in Lemma \ref{lem1} and equation \eqref{eq:def-tau1}, we get
$$
\big(e^{-z^{-1}\cdot\tilde\partial_{\bm t}}-1\big)(\partial\log\tau+s_1)=0
\,.
$$
On the other hand, the kernel of $e^{-z^{-1}\cdot\tilde\partial_{\bm t}}-1$ 
contains all functions of $x$, which are independent of ${\bm t}$.
We thus conclude that $\tau$ must satisfy
\begin{equation}\label{eq:partial-tau}
\partial\log\tau
=
-s_1+f
\,,
\end{equation}
where $f$ is such that $\frac{\partial f}{\partial t_k}=0$ for every $k\geq1$,
and we may set it equal to $0$ without loss of generality.

\medskip

Next, we want to derive from \eqref{eq:def-tau1} an equation for $\frac{\partial}{\partial t_k}\log\tau$, $k\geq1$.
We introduce the following differential operator
\begin{equation}\label{eq:N}
N(z)=\frac{\partial}{\partial z}-\sum_{k=1}^{\infty}z^{-k-1}\frac{\partial}{\partial t_k}\,.
\end{equation}
\begin{lemma}\label{lem:N}
$N(z) \big(e^{-z^{-1}\cdot\tilde\partial_{\bm t}} f\big)=0$ for every $f\in\mc F$.
\end{lemma}
\begin{proof}
We have
$$
\frac{\partial}{\partial z}e^{-z^{-1}\cdot\tilde\partial_{\bm t}}
=
\frac{\partial}{\partial z}\exp\{-\sum_k\frac1{kz^k}\frac{\partial}{\partial t_k}\}
=
\Big(\sum_{k=1}^\infty z^{-k-1}\frac{\partial}{\partial t_k}\Big)\circ
e^{-z^{-1}\cdot\tilde\partial_{\bm t}}
\,.
$$
\end{proof}
We then apply the operator $N(z)$ to equation \eqref{eq:def-tau1} and use Lemma \ref{lem:N} to get
\begin{equation}\label{eq:tauz}
N(z)\log(S(z))=\sum_{k=1}^\infty z^{-k-1}\frac{\partial}{\partial t_k}\log\tau
\,.
\end{equation}
Therefore, equation \eqref{20170912:eq1} is equivalent to \eqref{eq:tauz}
which, by looking at the various powers of $z$, is equivalent to the following system of equations:
\begin{equation}\label{eq:dtk-tau}
\frac{\partial}{\partial t_k}\log\tau
=
\Res_zz^kN(z)\log(S(z))=: g_k
\,.
\end{equation}

\medskip

In order to prove Theorem \ref{20170912:cor1},
we only need to show that the system of equations consisting of \eqref{eq:partial-tau} 
and \eqref{eq:dtk-tau} admits a solution $\phi:=\log\tau\in\mc F$, unique up to an additive constant.
Equivalently, since by assumption $\mc F$ is integrable,
we only need to prove that that system is compatible,  i.e.
\begin{equation}\label{eq:proof-tau1}
\partial g_k
=
-\frac{\partial s_1}{\partial t_k}
\,\,\text{ for all } k\geq1
\,,
\end{equation}
and
\begin{equation}\label{eq:proof-tau2}
\frac{\partial g_k}{\partial t_h}
=
\frac{\partial g_h}{\partial t_k}
\,\,,\,\,\,\,
\text{ for every } h,k\geq1
\,.
\end{equation}

Applying $N(z)$ to both sides of the equation in Lemma \ref{lem1}(b) and taking the coefficient of $z^{-k-1}$,
we immediately get \eqref{eq:proof-tau1} thanks to Lemma \ref{lem:N}.
Next, we prove equation \eqref{eq:proof-tau2}. For this, we shall need the following result.
\begin{lemma}\label{lem2}
The following identities hold:
\begin{enumerate}[(a)]
\item
$(e^{-w^{-1}\cdot\tilde\partial_{\bm t}}-1)\log S(z)=
(e^{-z^{-1}\cdot\tilde\partial_{\bm t}}-1)\log S(w)$.
\item
$N(z)N(w)\log S(z)=N(z)N(w)\log S(w)$.
\end{enumerate}
\end{lemma}
\begin{proof}
Setting $x'=x$ and $t_k'=t_k-\frac1{kw_1^k}-\frac1{kw_2^k}$, we have, as in \eqref{eq:exp},
\begin{equation}\label{20170910:eq7}
e^{z\cdot{\bm t}}e^{-z\cdot{\bm t'}}
=
\frac{w_1w_2}{w_2-w_1}\left(\iota_{w_1}(w_1-z)^{-1}-\iota_{w_2}(w_2-z)^{-1}\right)
\,.
\end{equation}
Applying the bilinear identity \eqref{eq:bilinear} and using equation \eqref{20170910:eq7} we get
\begin{align*}
&0=\res_z\left(\left(\iota_{w_1}(w_1-z)^{-1}-\iota_{w_2}(w_2-z)^{-1}\right)
S(z)e^{-w_1^{-1}\cdot\tilde\partial_{\bm t}}e^{-w_2^{-1}\cdot\tilde\partial_{\bm t}}(S^*)^{-1}(-z)\right)
\\
&=\left(S(z)e^{-w_1^{-1}\cdot\tilde\partial_{\bm t}}e^{-w_2^{-1}\cdot\tilde\partial_{\bm t}}(S^*)^{-1}(-z)\right)_-\Big|_{z=w_2} \\
&\qquad -\left(S(z)e^{-w_1^{-1}\cdot\tilde\partial_{\bm t}}e^{-w_2^{-1}\cdot\tilde\partial_{\bm t}}(S^*)^{-1}(-z)\right)_-\Big|_{z=w_1}
\\
&=S(w_2)e^{-w_1^{-1}\cdot\tilde\partial_{\bm t}}\frac{1}{S(w_2)}
-S(w_1)e^{-w_2^{-1}\cdot\tilde\partial_{\bm t}}\frac{1}{S(w_1)}
\,.
\end{align*}
In the second equality we used the second equation in \eqref{eq:positive}.
The $-$ sign here means that we need to take negative powers in $z$ \emph{before} 
substituting $z=w_2$ in the first parenthesis
and $z=w_1$ in the second.
In the third equality we used Lemma \ref{lem1}(a).
We rewrite the above identity as
$$
\frac{e^{-w^{-1}\cdot\tilde\partial_{\bm t}}S(z)}{S(z)}=\frac{e^{-z^{-1}\cdot\tilde\partial_{\bm t}}S(w)}{S(w)}
$$
and taking logarithm of both sides, we get claim (a).
Moreover, applying $N(z)N(w)$ to both sides of (a) and using Lemma \ref{lem:N}, we get claim (b).
\end{proof}

We can finally prove equation \eqref{eq:proof-tau2}.
We have, by the definition \eqref{eq:dtk-tau} of $g_k$ and Lemma \ref{lem2}(b),
\begin{align*}
&\sum_{h,k=1}^\infty\left(\frac{\partial g_h}{\partial t_k}-\frac{\partial g_k}{\partial t_h}\right)z^{-h-1}w^{-k-1}
\\
&=\sum_{k=1}^\infty w^{-k-1}\frac{\partial}{\partial t_k}\left(\sum_{h=1}^{\infty}g_hz^{-h-1}\right)
-\sum_{h=1}^\infty z^{-h-1}\frac{\partial}{\partial t_h}\left(\sum_{k=1}^{\infty}g_kw^{-k-1}\right)
\\
&=-N(w)N(z)\log S(z)+N(z)N(w)\log S(w)=0
\,.
\end{align*}
This proves \eqref{eq:proof-tau2} and concludes the proof of Theorem \ref{20170912:cor1}.

\subsection{Tau-functions of BKP and CKP type}

We next want to construct tau-functions when $\mc Z_L=2\mb Z_{\geq0}+1$.
The generic situation of this is provided by the BKP and CKP hierarchies, studied in \cite{DJKM81,DJKM83,DMH09,CW13,KZ20}
and other papers.
Throughout this section we assume that 
the algebra of functions $\mc F$, in the space variable $x$ and odd time variables $t_k$, $k\in\mc Z_L$,
is integrable (cf. Definition \ref{def:integrable}).
Recall the oscillating/antioscillating functions $e^{\pm z\cdot\bm t}=e^{\pm(zx+\sum_{k\in\mc Z_L}z^kt_k)}$,
and introduce the operator
$$
e^{-2z^{-1}\cdot\tilde{\partial}_{\bm t}}:\mc F\to\mc F((z^{-1}))
\,,
\text{ where }
\,
2z^{-1}\cdot\tilde{\partial}_{\bm t}=2\sum_{k\in\mc Z_L}\frac{z^{-k}}{k}\frac{\partial}{\partial t_k}
\,.
$$
By Taylor expansion, it ``shifts'' the time variables $t_k\mapsto t_k-\frac2{kz^k}$, $k\in\mc Z_L$.
Fix a monic pseudodifferential operator 
$S(\partial)=1+\sum_{i=1}^{\infty}s_i\partial^{-i}\in1+\mc F[[\partial^{-1}]]\partial^{-1}$
of order $0$. 
Let $w(z)=S(z)e^{z\cdot{\bm t}}\in\mc F((z^{-1}))e^{z\cdot{\bm t}}$ be the corresponding oscillating function,
and let $w^\star(z)=(S^*)^{-1}(-z)e^{-z\cdot{\bm t}}$ be the corresponding adjoint anti-oscillating function.

The analogue of Theorem \ref{20170912:cor1} in the case when $\mc Z_L$ consists of odd integers
is the following.
\begin{theorem}\label{20170912:cor1-BC}
Suppose that the bilinear identity \eqref{eq:bilinear} holds,
and that $S(\partial)$ satisfies the following condition
\begin{equation}\label{eq:SSstar}
S^*(\partial)=\partial^nS(\partial)^{-1}\partial^{-n}\,.
\end{equation}
for $n=0,$ or $1$.
Then, there exists a function $\tau\in\mc F$ such that
\begin{equation}\label{20170912:eq1-BC}
w(z)
=
\big(1+z^{-1}\mc E(z)\big)^{\!\frac{1-n}2}
\,\frac{\Gamma(z)\tau}{\tau}
\,,
\end{equation}
where 
$$
\Gamma(z)=\Gamma(x,{\bm t};z)=e^{z\cdot{\bm t}}e^{-2z^{-1}\cdot\tilde\partial_{\bm t}}:\mc F\to\mc F((z^{-1}))e^{z\cdot{\bm t}}
\,,
$$
and
\begin{equation}\label{eq:mcE}
\mc E(z)
=
(e^{-2z^{-1}\cdot\tilde\partial_{\bm t}}-1)\partial\log\tau
\,\in\mc F[[z^{-1}]]z^{-1}
\,.
\end{equation}
Moreover, the function $\tau\in\mc F$ solving \eqref{20170912:eq1} is unique up to multiplication by a function of $x$
(constant in $\bm t$).
\end{theorem}
\begin{lemma}\label{lem:20210714}
Let $L(\partial)\in\mc V((\partial^{-1}))$ be as in \eqref{eq:L}, satisfying the constraint
\begin{equation}\label{eq:LLstar}
L^*(\partial)=(-1)^N\partial^nL(\partial)\partial^{-n}
\,\,\,\, \text{ for } n=0 \text{ or } 1
\,.
\end{equation}
Assume that the dressing equation \eqref{eq:dress} has solution. Then, there exists
$S(\partial)\in1+\mc F[[\partial^{-1}]]\partial^{-1}$
solving the dressing equation and such that condition \eqref{eq:SSstar} holds.
\end{lemma}
\begin{proof}
Let $\tilde S(\partial)\in1+\mc F[[\partial^{-1}]]\partial^{-1}$ be a solution to the dressing equation \eqref{eq:dress}. The constraint
\eqref{eq:LLstar} implies that
$$
L(\partial)=\tilde S(\partial) \partial^N \tilde S^{-1}(\partial)
=\big(\partial^{-n}(\tilde S^*)^{-1}(\partial)\partial^n\big)\partial^N\big(\partial^{-n}(\tilde S^*)^{-1}(\partial)\partial^n\big)^{-1}
\,.
$$
Hence, by Theorem \ref{thm:sato}, we have
\begin{equation}\label{eq:K}
\partial^{-n}(\tilde S^*)^{-1}(\partial)\partial^n=\tilde S(\partial)K(\partial)\,,
\end{equation}
where $K(\partial)\in1+\mb F[[\partial^{-1}]]\partial^{-1}$.
Taking the adjoint of both sides of equation \eqref{eq:K} we get $K^*(\partial)=\partial^nK(\partial)\partial^{-n}=K(\partial)$.
Hence, the symbol of $K(\partial)$ satisfies $K(z)=K(-z)$.
Let $S(\partial)=\tilde S(\partial)C(\partial)$, with $C(\partial)\in1+\mb F[[\partial^{-1}]]\partial^{-1}$ be another solution
to the dressing equation \eqref{eq:dress}. To conclude the proof we need to show that there exists $C(\partial)$ such that
the constraint \eqref{eq:SSstar} holds. Replacing $S(\partial)=\tilde S(\partial)C(\partial)$ in \eqref{eq:SSstar}, we get
that it is equivalent to the condition
$$
C^*(\partial)\tilde S^*(\partial)=\partial^nC^{-1}(\partial)\tilde S^{-1}(\partial)\partial^{-n}
\!=C^{-1}(\partial)\partial^n\tilde S^{-1}(\partial)\partial^{-n}\!=C^{-1}(\partial)K^*(\partial)\tilde S^*(\partial)\,.
$$
In the last equality we used the adjoint of equation \eqref{eq:K}. The claim follows from the above equation and the fact that
$C(z)C(-z)=K(z)$ has a solution in $1+\mb F[[z^{-1}]]z^{-1}$ since $K(z)=K(-z)$.
\end{proof}
As an immediate corollary of Theorem \ref{20170912:cor1-BC}, Lemma \ref{lem:20210714} and Corollary \ref{cor:equiv},
we get the following
\begin{corollary}\label{20170912:cor2-BCKP}
Let $L(\partial)\in\mc V((\partial^{-1}))$ be as in \eqref{eq:L}, satisfying the constraint \eqref{eq:LLstar}.
Let $\varphi=(\varphi_\alpha)_{\alpha\in I}$, $\varphi_\alpha\in\mc F$,
be a solution of the hierarchy of Lax equations \eqref{eq:lax-sol} such that $a_0(\varphi)=0$.
Then, there exists $\tau\in\mc F$
such that
\begin{equation}\label{eq:cor2-BC}
L(\varphi;\partial)=S(\partial)\partial^N S(\partial)^{-1}
\,\,\text{ with }\,\,
S(z)=\big(1+z^{-1}\mc E(z)\big)^{\!\frac{1-n}2}\frac{e^{-2z^{-1}\cdot\tilde\partial_{\bm t}}\tau}{\tau}
\,.
\end{equation}
\end{corollary}
\begin{proof}
By Theorem \ref{thm:sato}(b), there exists $S(\partial)\in1+\mc F[[\partial^{-1}]]\partial^{-1}$
satisfying the dressing equation \eqref{eq:dress}.
Since $L(\partial)$ satisfies the constraint \eqref{eq:LLstar}, the corresponding $S(\partial)$
can be chosen so that \eqref{eq:SSstar} holds by Lemma \ref{lem:20210714}.
Hence by Theorem \ref{20170912:cor1-BC} there exists $\tau\in\mc F$
such that \eqref{eq:cor2-BC} holds.
%
\end{proof}

The remainder of this section will be devoted to the proof of Theorem \ref{20170912:cor1-BC}.
First note that equation \eqref{20170912:eq1-BC} defining $\tau$,
which we want to solve,
is obviously equivalent to the  the equation relating $S(z)$ and $\tau$ in \eqref{eq:cor2-BC}.
Applying the logarithm to both sides of the latter, we get, as we did in \eqref{eq:def-tau1},
\begin{equation}\label{eq:def-tau1-BC}
\log(S(z))
=
\big(e^{-2z^{-1}\cdot\tilde\partial_{\bm t}}-1\big)\log\tau
+
\frac{1-n}2 \log \big(1+z^{-1}\mc E(z)\big)
\,.
\end{equation}
Follow the line of reasoning used for the proof of Theorem \ref{20170912:cor1},
the proof of Theorem \ref{20170912:cor1-BC} will be obtained as follows.
From equation \eqref{eq:def-tau1-BC} we shall derive an (equivalent) system of partial differential equations
of type \eqref{20170721:eq2} for $\phi=\log\tau$,
depending on a function $f$ of $x$.
We will then prove that such a system is compatible. 
The claim of Theorem \ref{20170912:cor1-BC} will then be a consequence of 
the integrability assumption on $\mc F$.

\medskip

In analogy with \eqref{eq:mcE}, for $i\geq1$ we denote
\begin{equation}\label{eq:mcEi}
\mc E_i(z)
=
(e^{-2z^{-1}\cdot\tilde\partial_{\bm t}}-1)s_i
\,.
\end{equation}
In order to derive from \eqref{eq:def-tau1-BC} an equation for $\partial\log\tau$,
we need the following analogue of Lemma \ref{lem1}.
\begin{lemma}\label{lem1-BC}
\begin{enumerate}[a)]
\item
The following identity holds:
\begin{equation}\label{eq:lem1-BC-a}
S(z)e^{-2z^{-1}\cdot\tilde{\partial}_{\bm t}}(S^*)^{-1}(-z)
=
1-\frac12z^{-1}\mc E_1(z)
\,.
\end{equation}
\item
The following identity holds:
\begin{equation}\label{eq:lem1-BC-b}
\partial\log S(z)
=
-\frac12\mc E_1(z)
-\frac12 z^{-1}
\big(1-\frac12z^{-1}\mc E_1(z)\big)^{-1}
\big(
\mc E_2(z)+\partial\mc E_1(z)-s_1\mc E_1(z)-\frac12\mc E_1(z)^2
\big)
\,.
\end{equation}
\end{enumerate}
\end{lemma}
\begin{proof}
In analogy with the proof of Lemma \ref{lem1} we set $t_k'=t_k-\frac{2}{kw^k}$, $k\geq1$, 
in the bilinear identity \eqref{eq:bilinear} or, equivalently, \eqref{20170910:eq4a}.
A straightforward computation leads to
\begin{equation}\label{eq:exp-odd}
e^{z\cdot\bm t}
e^{-z\cdot{\bm t'}}=\big(1+\frac{z}{w}\big)\iota_w\big(1-\frac zw\big)^{-1}e^{z(x-x')}
\,.
\end{equation}
Hence, equation \eqref{20170910:eq4a} gives
$$
\Res_z\big(
(w+z)\iota_w(w-z)^{-1}
e^{z(x-x')}
S(x,\bm t;z)
e^{-2w^{-1}\cdot\tilde\partial_{\bm t}}
(S^*)^{-1}(x',\bm t;-z)
\big)
=0
\,,
$$
which, by the second equation in \eqref{eq:positive}, gives
\begin{equation}\label{20170910:eq4b-odd}
\Big(
(w+z)e^{z(x-x')}
S(x,\bm t;z)
e^{-2w^{-1}\cdot\tilde\partial_{\bm t}}
(S^*)^{-1}(x',\bm t;-z)
\Big)_-\Big|_{z=w}
=0
\,,
\end{equation}
where, as in \eqref{20170910:eq4b}, the $-$ means that we need to take the negative powers of $z$, before setting $z=w$.
We rewrite \eqref{20170910:eq4b-odd} as
\begin{equation}\label{20170910:eq4c-odd}
\begin{split}
& 2we^{w(x-x')}
S(x,\bm t;w)
e^{-2w^{-1}\cdot\tilde\partial_{\bm t}}
(S^*)^{-1}(x',\bm t;-w) \\
& =
\Big((w+z)
e^{z(x-x')}
S(x,\bm t;z)
e^{-2w^{-1}\cdot\tilde\partial_{\bm t}}
(S^*)^{-1}(x',\bm t;-z)
\Big)_+\Big|_{z=w}
\,.
\end{split}
\end{equation}
A straightforward computation shows that, if $S(z)=1+s_1z^{-1}+s_2z^{-2}+\dots$, then
\begin{equation}\label{eq:A}
(S^*)^{-1}(-z)
=
1-s_1z^{-1}+(s_1^2-s_2-s_1')z^{-2}+\dots
\,.
\end{equation}
Setting $x'=x$ in \eqref{20170910:eq4c-odd}, the RHS becomes
$$
\Big(
(w+z)S(z)
e^{-2w^{-1}\cdot\tilde\partial_{\bm t}}
(S^*)^{-1}(-z)
\Big)_+\Big|_{z=w}
=
2w
+
s_1
-
e^{-w^{-1}\cdot\tilde\partial_{\bm t}} s_1
=
2w-\mc E_1(w)
\,.
$$
Claim (a) follows.
Next, we apply $\frac{\partial}{\partial x}$ to both sides of \eqref{20170910:eq4c-odd}
and then we set $x'=x$. The LHS gives
$$
2w\big((w+\partial)S(w)\big)
e^{-2w^{-1}\cdot\tilde\partial_{\bm t}}
(S^*)^{-1}(-w) 
\,,
$$
which, by part (a), is
\begin{equation}\label{20210628:eq1}
2w^2
-
w\mc E_1(w)
+
2w
\big(1-\frac12w^{-1}\mc E_1(w)\big)
\partial\log S(w)
\,.
\end{equation}
On the other hand, the RHS of \eqref{20170910:eq4c-odd} gives, using \eqref{eq:A},
\begin{align*}
&\Big((w+z)
\big((z+\partial)S(z)\big)
e^{-2w^{-1}\cdot\tilde\partial_{\bm t}}
(S^*)^{-1}(-z)
\Big)_+\Big|_{z=w}
\\
& =
2w^2
-
2w
\mc E_1(w)
-\mc E_2(w)
-\partial\mc E_1(w)
+
e^{-2w^{-1}\cdot\tilde\partial_{\bm t}}
s_1^2
-
s_1
e^{-2w^{-1}\cdot\tilde\partial_{\bm t}}
s_1
\,,
\end{align*}
which can be rewritten as 
\begin{equation}\label{20210628:eq2}
2w^2
-
2w
\mc E_1(w)
-\mc E_2(w)
-\partial\mc E_1(w)
+
\mc E_1(w)^2+s_1\mc E_1(w)
\,,
\end{equation}
since, obviously, 
$e^{-2w^{-1}\cdot\tilde\partial_{\bm t}} s_1^2 = \big(e^{-2w^{-1}\cdot\tilde\partial_{\bm t}}s_1\big)^2=(\mc E_1(w)+s_1)^2$.
Equating \eqref{20210628:eq1} and \eqref{20210628:eq2},
a straightforward computation leads to \eqref{eq:lem1-BC-b}.
\end{proof}
So far we did not use the assumption that $S(\partial)$ satisfies the condition \eqref{eq:SSstar}.
If we do, the equation \eqref{eq:lem1-BC-b} greatly simplifies.
\begin{lemma}\label{lem1-BC-b}
Assume that $S(\partial)$ satisfies \eqref{eq:SSstar} for some $n\in\mb Z$.
Then, the following equations holds
\begin{enumerate}[(a)]
\item
$s_2=\frac12 s_1^2-\frac{n+1}{2}s_1'$.
\item
$\mc E_2(z)
=
\frac12\mc E_1(z)^2+s_1\mc E_1(z)-\frac{n+1}2\partial\mc E_1(z)$.
\item
Equation \eqref{eq:lem1-BC-b} reduces to
\begin{equation}\label{eq:lem1-BC-c}
\partial\log S(z)
=
-\frac12\mc E_1(z)
+
\frac{1-n}2
\partial \log \big(1-\frac12z^{-1}\mc E_1(z)\big)
\,.
\end{equation}
\end{enumerate}
\end{lemma}
\begin{proof}
Expanding the identity $S^*(\partial)\partial^n S(\partial) =\partial^{n}$ we get
$$
\partial^n+(2s_2+(n+1)s_1'-s_1^2)\partial^{n-2}+\dots=\partial^n\,.
$$
Claim (a) follows from this identity.
Claim (b) is an immediate consequence of (a) and the definition \eqref{eq:mcEi} of $\mc E_i(z)$.
Claim (c) follows from (b).
\end{proof}
Comparing equations \eqref{eq:def-tau1-BC} and \eqref{eq:lem1-BC-c}, we get, recalling \eqref{eq:mcE}
$$
\mc E(z)
+
\frac{1-n}2 \partial \log \big(1+z^{-1}\mc E(z)\big)
=
-\frac12\mc E_1(z)
+
\frac{1-n}2
\partial \log \big(1-\frac12z^{-1}\mc E_1(z)\big)
\,,
$$
which implies 
\begin{equation}\label{eq:mcEE1}
\mc E(z)=-\frac12\mc E_1(z)
\,.
\end{equation}
%
%
Since the kernel of $e^{-z^{-1}\cdot\tilde\partial_{\bm t}}-1$ 
contains all functions of $x$, which are independent of ${\bm t}$,
we then conclude, as in \eqref{eq:partial-tau}, that $\tau$ must satisfy
\begin{equation}\label{eq:partial-tau-BC}
\partial\log\tau
=
-\frac12 s_1+f
\,,
\end{equation}
where $f$ is a function constant in all $t_k$'s,
which, without loss of generality, we may set equal to $0$.

\medskip

Next, we want to derive from \eqref{eq:def-tau1-BC} an equation for $\frac{\partial}{\partial t_k}\log\tau$, $k\geq1$.
In analogy with \eqref{eq:N},
we introduce the differential operator
\begin{equation}\label{eq:N-BC}
N(z)=\frac{\partial}{\partial z}-2\sum_{k\in\mc Z_L}^{\infty}z^{-k-1}\frac{\partial}{\partial t_k}\,,
\end{equation}
and we observe that the analogue of Lemma \ref{lem:N} still holds:
\begin{equation}\label{lem:N-BC}
N(z) \big(e^{-2z^{-1}\cdot\tilde\partial_{\bm t}} f\big)=0\,\text{ for every } f\in\mc F
\,.
\end{equation}
We then apply the operator $N(z)$ to equation \eqref{eq:def-tau1-BC} and use equation \eqref{lem:N-BC} to get
$$
N(z)\Big(
\log S(z)
-
\frac{1-n}2 
\log \big(1+z^{-1}\mc E(z)\big)
\Big)
=
2\sum_{k\in\mc Z_L}^{\infty}z^{-k-1}\frac{\partial}{\partial t_k}
\log\tau
\,.
$$
By looking at the various powers of $z$, this is equivalent to the following system of equations:
\begin{equation}\label{eq:dtk-tau-BC}
\frac{\partial}{\partial t_k}\log\tau
=
\frac12
\Res_zz^k
N(z)\Big(
\log S(z)
-
\frac{1-n}2 
\log \big(1+z^{-1}\mc E(z)\big)
\Big)
=: g_k
\,,\,\,
k\in\mc Z_L
\,.
\end{equation}

\medskip

In order to prove Theorem \ref{20170912:cor1-BC},
by the integrability assumption on $\mc F$,
it remains to show that the system of equations consisting of \eqref{eq:partial-tau-BC} 
and \eqref{eq:dtk-tau-BC} is compatible,  i.e.
\begin{equation}\label{eq:proof-tau1-BC}
\partial g_k
=
-\frac12 \frac{\partial s_1}{\partial t_k}
\,\,\text{ for every } k\in\mc Z_L=2\mb Z_{\geq0}+1
\,,
\end{equation}
and
\begin{equation}\label{eq:proof-tau2-BC}
\frac{\partial g_k}{\partial t_h}
=
\frac{\partial g_h}{\partial t_k}
\,\,,\,\,\,\,
\text{ for every } h,k\in\mc Z_L
\,.
\end{equation}

Applying $N(z)$ to both sides of equation \eqref{eq:lem1-BC-c} and taking the coefficient of $z^{-k-1}$,
we immediately get \eqref{eq:proof-tau1-BC}, thanks to equation \eqref{lem:N-BC}.
We are left to prove equation \eqref{eq:proof-tau2-BC}. 
We can write \eqref{eq:proof-tau2-BC} as a formal power series
by multiplying both sides by $2z^{-h-1}w^{-k-1}$ and summing over $h,k\in\mc Z_L$.
As a result, recalling the definition \eqref{eq:N-BC} of $N(z)$,
equation \eqref{eq:proof-tau2-BC} becomes
\begin{align*}
& 0
=
2\sum_{h,k\in\mc Z_L}z^{-h-1}w^{-k-1}
\Big(\frac{\partial g_k}{\partial t_h}
-
\frac{\partial g_h}{\partial t_k}
\Big) \\
& =
2
\sum_{h\in\mc Z_L}
z^{-h-1}
\frac{\partial}{\partial t_h}
\sum_{k\in\mc Z_L}
w^{-k-1}
g_k
-
2
\sum_{k\in\mc Z_L}
w^{-k-1}
\frac{\partial}{\partial t_k}
\sum_{h\in\mc Z_L}
z^{-h-1}
g_h \\
\end{align*}
Notice that $2\sum_{h\in\mc Z_L} z^{-h-1}\frac{\partial}{\partial t_h}$ can be replaced by $-N(z)$,
since it acts on a functions which is constant in $z$.
Moreover, recalling the definition \eqref{eq:dtk-tau-BC} of $g_k$,
we have that 
$\sum_{k\in\mc Z_L} w^{-k-1} g_k$ coincides with the even part (i.e. even powers of $w$) 
of $N(w)\Big(\log S(w)-\frac{1-n}2 \log \big(1+w^{-1}\mc E(w)\big)\Big)$.
Hence, recalling \eqref{eq:mcEE1}, equation \eqref{eq:proof-tau2-BC} is translated in the following equation
\begin{equation}\label{eq:proof-tau2-BCb}
\begin{split}
& N(z)
\Big(
N(w)\log S(w)
-
\frac{1-n}2 
N(w)\log \big(1-\frac12w^{-1}\mc E_1(w)\big)
\Big)_{\text{even}} 
\\
& =
N(w)
\Big(N(z)
\log S(z)
-
\frac{1-n}2 
N(z)\log \big(1-\frac12z^{-1}\mc E_1(z)\big)
\Big)_{\text{even}}
\,,
\end{split}
\end{equation}
where the index ``even'' means that we take only the terms with even powers of $w$ (resp. $z$) in the LHS (resp. RHS).
In order to complete the proof of Theorem \ref{20170912:cor1-BC},
we are left to prove equation \eqref{eq:proof-tau2-BCb}.
We shall prove in the following two subsections equation \eqref{eq:proof-tau2-BCb}
separately for $n=0$ and $n=1$.

\subsection{Proof of equation \eqref{eq:proof-tau2-BCb} for $n=0$}

Setting $x'=x$ and $t_k'=t_k-\frac2{kw_1^k}-\frac2{kw_2^k}$, we have, as in \eqref{eq:exp-odd},
\begin{equation}
\begin{split}\label{20170910:eq7-odd}
&e^{z\cdot{\bm t}}e^{-z\cdot{\bm t'}}
=
(1+\frac{z}{w_1})(1+\frac{z}{w_2})\iota_{w_1}(1-\frac{z}{w_1})^{-1}\iota_{w_2}(1-\frac{z}{w_2})^{-1}
\\
&=
1+2\frac{w_1+w_2}{w_1-w_2}\left(-w_1\iota_{w_1}(w_1-z)^{-1}+w_2\iota_{w_2}(w_2-z)^{-1}\right)\,.
\end{split}
\end{equation}
By the bilinear identity \eqref{eq:bilinear} and equation \eqref{20170910:eq7-odd} we get
\begin{equation}\label{eq:x}
\begin{split}
& \res_z\left(S(z)e^{-2w_1^{-1}\cdot\tilde\partial_{\bm t}}e^{-2w_2^{-1}\cdot\tilde\partial_{\bm t}}(S^*)^{-1}(-z)\right)
\\
& =
2\frac{w_1+w_2}{w_1-w_2}\res_z
\Big(
\left(w_1\iota_{w_1}(w_1-z)^{-1}-w_2\iota_{w_2}(w_2-z)^{-1}\right)\times \\
&\qquad\qquad\times
S(z)e^{-2w_1^{-1}\cdot\tilde\partial_{\bm t}}e^{-2w_2^{-1}\cdot\tilde\partial_{\bm t}}(S^*)^{-1}(-z)
\Big)
\,.
\end{split}
\end{equation}
Recalling that $S(z)=1+s_1z^{-1}+\dots$ and $(S^*)^{-1}(-z)=1-s_1z^{-1}+\dots$, the LHS of \eqref{eq:x} is
\begin{equation}\label{eq:aaa}
-\big(e^{-2w_1^{-1}\cdot\tilde\partial_{\bm t}}e^{-2w_2^{-1}\cdot\tilde\partial_{\bm t}}-1\big)s_1
\,.
\end{equation}
Using the second equation in \eqref{eq:positive} we can rewrite the RHS of \eqref{eq:aaa} as
\begin{equation}\label{20210628:eq6}
\begin{split}
&2\frac{w_1+w_2}{w_1-w_2}\left(
w_1\left(
S(z)e^{-2w_1^{-1}\cdot\tilde\partial_{\bm t}}e^{-2w_2^{-1}\cdot\tilde\partial_{\bm t}}(S^*)^{-1}(-z)
\right)_-\Big|_{z=w_1}
\right.
\\
&
\qquad \left.-
w_2\left(
S(z)e^{-2w_1^{-1}\cdot\tilde\partial_{\bm t}}e^{-2w_2^{-1}\cdot\tilde\partial_{\bm t}}(S^*)^{-1}(-z)
\right)_-\Big|_{z=w_2}
\right)
\,.
\end{split}
\end{equation}
The $-$ sign means that we need to take negative powers in $z$ before substituting $z=w_1$ in the first parenthesis
and $z=w_2$ in the second.
We next observe that
\begin{equation}\label{eq:bbb}
\begin{split}
&
\left(
S(z)e^{-2w_1^{-1}\cdot\tilde\partial_{\bm t}}e^{-2w_2^{-1}\cdot\tilde\partial_{\bm t}}(S^*)^{-1}(-z)
\right)_-\Big|_{z=w_1} \\
&=
S(w_1)e^{-2w_1^{-1}\cdot\tilde\partial_{\bm t}}e^{-2w_2^{-1}\cdot\tilde\partial_{\bm t}}(S^*)^{-1}(-w_1)
-
1
\,,
\end{split}
\end{equation}
and similarly
\begin{equation}\label{eq:ccc}
\begin{split}
&
\left(
S(z)e^{-2w_1^{-1}\cdot\tilde\partial_{\bm t}}e^{-2w_2^{-1}\cdot\tilde\partial_{\bm t}}(S^*)^{-1}(-z)
\right)_-\Big|_{z=w_2} \\
&=
S(w_2)e^{-2w_1^{-1}\cdot\tilde\partial_{\bm t}}e^{-2w_2^{-1}\cdot\tilde\partial_{\bm t}}(S^*)^{-1}(-w_2)
-
1
\,.
\end{split}
\end{equation}
Combining equations \eqref{eq:x}, \eqref{eq:aaa}, \eqref{20210628:eq6}, \eqref{eq:bbb} and \eqref{eq:ccc}, we get
\begin{align*}
& -\big(e^{-2w_1^{-1}\cdot\tilde\partial_{\bm t}}e^{-2w_2^{-1}\cdot\tilde\partial_{\bm t}}-1\big)s_1 \\
& =
2\frac{w_1+w_2}{w_1-w_2}
\left(
w_1
S(w_1)e^{-2w_1^{-1}\cdot\tilde\partial_{\bm t}}e^{-2w_2^{-1}\cdot\tilde\partial_{\bm t}}(S^*)^{-1}(-w_1)
\right.
\\
&
\qquad \left.
-
w_2
S(w_2)e^{-2w_1^{-1}\cdot\tilde\partial_{\bm t}}e^{-2w_2^{-1}\cdot\tilde\partial_{\bm t}}(S^*)^{-1}(-w_2)
\right)
-2(w_1+w_2)
\,.
\end{align*}
Using equation \eqref{eq:lem1-BC-a}, the above equation becomes, after some manipulations,
\begin{equation}\label{eq:dan1}
\begin{split}
& 
w_1
S(w_1)
\frac{1-\frac12 w_1^{-1} e^{-2w_2^{-1}\cdot\tilde\partial_{\bm t}} \mc E_1(w_1)}{e^{-2w_2^{-1}\cdot\tilde\partial_{\bm t}} S(w_1)}
-
w_2
S(w_2)
\frac{1-\frac12 w_2^{-1} e^{-2w_1^{-1}\cdot\tilde\partial_{\bm t}} \mc E_1(w_2)}{e^{-2w_1^{-1}\cdot\tilde\partial_{\bm t}} S(w_2)} \\
& =
(w_1-w_2)
\Big(
1
-
\frac12 (w_1+w_2)^{-1}
\big(e^{-2w_1^{-1}\cdot\tilde\partial_{\bm t}}e^{-2w_2^{-1}\cdot\tilde\partial_{\bm t}}-1\big)s_1 
\Big)\,.
\end{split}
\end{equation}
In analogy with \eqref{eq:mcEi}, we let
$$
\mc E_1(z,w)
=
\big(e^{-2z^{-1}\cdot\tilde\partial_{\bm t}}e^{-2w^{-1}\cdot\tilde\partial_{\bm t}}-1\big)s_1 
\,,
$$
and we observe that
\begin{equation}\label{20210706:eq1}
e^{-2w_2^{-1}\cdot\tilde\partial_{\bm t}} \mc E_1(w_1)
=
\mc E_1(w_1,w_2)-\mc E_1(w_2)
\,.
\end{equation}
Hence, equation \eqref{eq:dan1} becomes
\begin{equation}\label{eq:dan3}
\begin{split}
& 
\big(
w_1-\frac12  e^{-2w_2^{-1}\cdot\tilde\partial_{\bm t}} \mc E_1(w_1)\big)
\frac{S(w_1)}{e^{-2w_2^{-1}\cdot\tilde\partial_{\bm t}} S(w_1)} \\
& -
\big(
w_2-\frac12  e^{-2w_1^{-1}\cdot\tilde\partial_{\bm t}} \mc E_1(w_2)\big)
\frac{S(w_2)}{e^{-2w_1^{-1}\cdot\tilde\partial_{\bm t}} S(w_2)} \\
& =
(w_1-w_2)
\Big(
1
-
\frac12 (w_1+w_2)^{-1} \mc E_1(w_1,w_2)
\Big)
\,.
\end{split}
\end{equation}
Next, we apply $e^{2w_2^{-1}\cdot\tilde\partial_{\bm t}}$ to both sides of \eqref{eq:dan3} to get
\begin{equation}\label{eq:dan4}
\begin{split}
& 
\big(
w_1-\frac12  \mc E_1(w_1)\big)
\frac{e^{2w_2^{-1}\cdot\tilde\partial_{\bm t}}S(w_1)}{ S(w_1)} \\
& -
\big(e^{2w_2^{-1}\cdot\tilde\partial_{\bm t}}S(w_2)
\big)
e^{-2w_1^{-1}\cdot\tilde\partial_{\bm t}} e^{2w_2^{-1}\cdot\tilde\partial_{\bm t}}
\frac{w_2-\frac12\mc E_1(w_2)}{S(w_2)} \\
& =
(w_1-w_2)
\Big(
1
-
\frac12 (w_1+w_2)^{-1} (\mc E_1(w_1)-\mc E_1(-w_2))
\Big)
\,,
\end{split}
\end{equation}
where in the RHS we used the identity
$$
e^{2w_2^{-1}\cdot\tilde\partial_{\bm t}}\mc E_1(w_1,w_2)
=e^{-2w_1^{-1}\cdot\tilde\partial_{\bm t}}s_1-e^{2w_2^{-1}\cdot\tilde\partial_{\bm t}}s_1
=\mc E_1(w_1)-\mc E_1(-w_2)\,.
$$
Recall that, by assumption, $S(\partial)$ satisfies the condition $(S^{*})^{-1}(\partial)=S(\partial)$ (cf. \eqref{eq:SSstar}).
Hence, by this assumption and equation \eqref{eq:lem1-BC-a} we have
$$
e^{2w_2^{-1}\cdot\tilde\partial_{\bm t}}S(w_2)
=e^{2w_2^{-1}\cdot\tilde\partial_{\bm t}}(S^{*})^{-1}(w_2)
=\frac{1+\frac12w_2^{-1}\mc E_1(-w_2)}{S(-w_2)}\,.
$$
Hence,
\begin{equation}
\begin{split}\label{eq:dan6}
&\big(e^{2w_2^{-1}\cdot\tilde\partial_{\bm t}}S(w_2)
\big)
e^{-2w_1^{-1}\cdot\tilde\partial_{\bm t}} e^{2w_2^{-1}\cdot\tilde\partial_{\bm t}}
\frac{w_2-\frac12\mc E_1(w_2)}{S(w_2)}
\\
&=
\frac{1+\frac12w_2^{-1}\mc E_1(-w_2)}{S(-w_2)}e^{-2w_1^{-1}\cdot\tilde\partial_{\bm t}} e^{2w_2^{-1}\cdot\tilde\partial_{\bm t}}
\big(w_2e^{-2w_2^{-1}\cdot\tilde\partial_{\bm t}}S(-w_2)\big)
\\
&=
\frac{w_2+\frac12\mc E_1(-w_2)}{S(-w_2)}e^{-2w_1^{-1}\cdot\tilde\partial_{\bm t}}
S(-w_2)\,.
\end{split}
\end{equation}
Substituting \eqref{eq:dan6} in the LHS of \eqref{eq:dan4} and replacing $w_2$ by $-w_2$ we arrive at the identity
\begin{equation}\label{eq:dan4b}
\begin{split}
& 
\big(
w_1-\frac12  \mc E_1(w_1)\big)
\frac{e^{-2w_2^{-1}\cdot\tilde\partial_{\bm t}}S(w_1)}{ S(w_1)}+
\big(w_2-\frac12\mc E_1(w_2)
\big)
\frac{e^{-2w_1^{-1}\cdot\tilde\partial_{\bm t}}S(w_2) }{S(w_2)}\\
& =
(w_1+w_2)
\Big(
1
-
\frac12 (w_1-w_2)^{-1} (\mc E_1(w_1)-\mc E_1(w_2))
\Big)
\,.
\end{split}
\end{equation}
Next, we multiply both sides of the identities \eqref{eq:dan3} and \eqref{eq:dan4b} to get
\begin{equation}\label{eq:dan7}
\begin{split}
& 
\big(
w_1-\frac12  e^{-2w_2^{-1}\cdot\tilde\partial_{\bm t}} \mc E_1(w_1)\big)
\big(
w_1-\frac12  \mc E_1(w_1)\big)
\\
&+
\big(
w_1-\frac12  e^{-2w_2^{-1}\cdot\tilde\partial_{\bm t}} \mc E_1(w_1)\big)
\big(w_2-\frac12\mc E_1(w_2)
\big)
\frac{S(w_1)}{e^{-2w_2^{-1}\cdot\tilde\partial_{\bm t}} S(w_1)}
\cdot
\frac{e^{-2w_1^{-1}\cdot\tilde\partial_{\bm t}}S(w_2) }{S(w_2)}
\\
& -
\big(
w_2-\frac12  e^{-2w_1^{-1}\cdot\tilde\partial_{\bm t}} \mc E_1(w_2)\big)
\big(
w_1-\frac12  \mc E_1(w_1)\big)
\frac{e^{-2w_2^{-1}\cdot\tilde\partial_{\bm t}}S(w_1)}{ S(w_1)}
\cdot
\frac{S(w_2)}{e^{-2w_1^{-1}\cdot\tilde\partial_{\bm t}} S(w_2)}
\\
&-
\big(
w_2-\frac12  e^{-2w_1^{-1}\cdot\tilde\partial_{\bm t}} \mc E_1(w_2)\big)
\big(w_2-\frac12\mc E_1(w_2)
\big)
\\
& =
\Big(
w_1+w_2
-
\frac12 \mc E_1(w_1,w_2)
\Big)
\Big(
w_1-w_2
-
\frac12 (\mc E_1(w_1)-\mc E_1(w_2))
\Big)
\,.
\end{split}
\end{equation}
Using \eqref{20210706:eq1}, it is straightforward to check that
\begin{align*}
& 
\big(
w_1-\frac12  e^{-2w_2^{-1}\cdot\tilde\partial_{\bm t}} \mc E_1(w_1)\big)
\big(
w_1-\frac12  \mc E_1(w_1)\big)
\\
& -
\big(
w_2-\frac12  e^{-2w_1^{-1}\cdot\tilde\partial_{\bm t}} \mc E_1(w_2)\big)
\big(w_2-\frac12\mc E_1(w_2)
\big)
\\
& =
\Big(
w_1+w_2
-
\frac12 \mc E_1(w_1,w_2)
\Big)
\Big(
w_1-w_2
-
\frac12 (\mc E_1(w_1)-\mc E_1(w_2))
\Big)
\,.
\end{align*}
Then, equation \eqref{eq:dan7} becomes
\begin{equation}\label{eq:dan8}
\begin{split}
&
\left(\frac{S(w_1)}{e^{-2w_2^{-1}\cdot\tilde\partial_{\bm t}} S(w_1)}
\cdot
\frac{e^{-2w_1^{-1}\cdot\tilde\partial_{\bm t}}S(w_2) }{S(w_2)}
\right)^2\\
&=
\frac{1-\frac12  w_1^{-1}\mc E_1(w_1)}{e^{-2w_2^{-1}\cdot\tilde\partial_{\bm t}} \big(
1-\frac12 w_1^{-1} \mc E_1(w_1)\big)}
\cdot
\frac{e^{-2w_1^{-1}\cdot\tilde\partial_{\bm t}} \big(
1-\frac12  w_2^{-1}\mc E_1(w_2)\big)}
{1-\frac12w_2^{-1}\mc E_1(w_2)}
\,.
\end{split}
\end{equation}
Taking logarithm of both sides of \eqref{eq:dan8} and dividing by $2$ we arrive at
\begin{equation}\label{eq:dan9}
\begin{split}
&
\left(e^{-2w_1^{-1}\cdot\tilde\partial_{\bm t}}-1\right)
\left(\log S(w_2)-\frac12\log \big(1-\frac12w_2^{-1}\mc E_1(w_2)\big)\right)
\\
&=
\left(e^{-2w_2^{-1}\cdot\tilde\partial_{\bm t}}-1\right)
\left(\log S(w_1)-\frac12\log \big(1-\frac12w_1^{-1}\mc E_1(w_1)\big)\right)
\,.
\end{split}
\end{equation}
Finally, we apply $N(w_1)N(w_2)$ to both sides of \eqref{eq:dan9} and use \eqref{lem:N-BC} to get
\begin{equation}\label{eq:dan10}
\begin{split}
&
N(w_1)N(w_2)
\left(\log S(w_2)-\frac12\log \big(1-\frac12w_2^{-1}\mc E_1(w_2)\big)\right)
\\
&=
N(w_1)N(w_2)
\left(\log S(w_1)-\frac12\log \big(1-\frac12w_1^{-1}\mc E_1(w_1)\big)\right)
\,.
\end{split}
\end{equation}
Equation \eqref{eq:dan10} implies \eqref{eq:proof-tau2-BCb} for $n=0$.

\subsection{Proof of equation \eqref{eq:proof-tau2-BCb} for $n=1$}

In \cite{DJKM83} the following bilinear identity for the BKP hierarchy is used
\begin{equation}\label{bilBKP}
\res_{z}\big(w(x,\bm t;,z)w(x',\bm t',-z)z^{-1}\big)=1
\,.
\end{equation}
\begin{lemma}
Let us assume that the bilinear identity for the BKP \eqref{bilBKP} holds. Then $S(\partial)$ satisfies the condition \eqref{eq:SSstar} with $n=1$ and the bilinear identity \eqref{eq:bilinear} holds.
Conversely, assume that $S(\partial)$ satisfies the condition \eqref{eq:SSstar} with $n=1$ and $w(z)=S(z)e^{z\cdot\bm t}$
satisfies the the bilinear identity \eqref{eq:bilinear}. Then $w(z)$ satisfies the bilinear identity for the BKP \eqref{bilBKP}.
\end{lemma}
\begin{proof}
Let us apply $\partial_x^m$, $m\geq0$, to \eqref{bilBKP} and set $x'=x$ and $t_k'=t_k$, $k\in\mc Z_L$. Then, we have
\begin{equation}\label{20210706:eq2}
\res_z\big( ((z+\partial)^mS(z))S(-z)z^{-1}\big)=\delta_{m,0}\,.
\end{equation}
Similarly to the proof of Lemma \ref{20170726:lem}, with $A(\partial)=S(\partial)$ and $B^*(\partial)=-S(\partial)\partial^{-1}$,
the identities \eqref{20210706:eq2} for $m>0$ imply $S(\partial)\partial^{-1}\circ S^*(\partial)\in\mc F[\partial]\partial^{-1}$.
Since $S(\partial)\in1+\mc F[[\partial^{-1}]]\partial^{-1}$ we then have $S(\partial)\partial^{-1}\circ S^*(\partial)=a\partial^{-1}$, for some $a\in\mc F$. From \eqref{20210706:eq2} with $m=0$ we get that $a=1$ thus showing that $S(\partial)$ satisfies the
condition \eqref{eq:SSstar} with $n=1$. In particular, we have
\begin{equation}\label{eq:SBKP}
(S^*)^{-1}(-z)=-(-z+\partial)S(-z)z^{-1}
\,.
\end{equation}
Next, let us apply $\partial_{x'}$ to both sides of \eqref{bilBKP}. We get
\begin{align*}
&0=\res_z\big(w(x,\bm t;,z)\partial_{x'}w(x',\bm t',-z)z^{-1}\big)
\\
&=
\res_z\big( w(x,\bm t;,z)(-z+\partial_{x'})S(-z)z^{-1}e^{-zx'-z\cdot\bm t'}\big)
\\
&=-\res_z\big( w(x,\bm t;,z)w^\star(x',\bm t',z)\big)
\,.
\end{align*}
In the last identity we used \eqref{eq:SBKP}. This shows that $w(z)$ satisfies \eqref{eq:bilinear}.

Conversely, let us assume that $S(\partial)$ satisfies \eqref{eq:SBKP} (that is condition \eqref{eq:SSstar} with $n=1$)
and $w(z)=S(z)e^{z\cdot\bm t}$ satisfies the bilinear identity \eqref{eq:bilinear}, and let
\begin{equation}\label{eq:f}
f=\res_z(w(x,\bm t,z)w(x',\bm t',-z)z^{-1})\,.
\end{equation}
By \eqref{eq:SBKP}, we have
$$
w^\star(z)
=
(S^*)^{-1}(-z)e^{-z\cdot\bm t}
=
-(-z+\partial)S(-z)z^{-1}e^{-z\cdot\bm t}
=
-\partial w(-z)z^{-1}
\,.
$$
Hence,
\begin{equation}\label{20210706:eq3}
\partial_{x'}f=-\res_z(w(x,\bm t,z)\partial_{x'}w(x',\bm t',-z)z^{-1})
=-\res_z(w(x,\bm t,z)w^\star(x',\bm t',z))=0
\,.
\end{equation}
Moreover, by Corollary \eqref{cor:equiv} and \eqref{eq:linear} we have ($k\in\mc Z_L$)
\begin{align*}
&\frac{\partial f}{\partial t'_k}=\res_z\left(w(x,\bm t,z)\frac{\partial w(x',\bm t',-z)}{\partial t_k'}z^{-1}\right)
\\
&=\res_z\left(w(x,\bm t,z)B_k(\partial_{x'})w(x',\bm t',-z)z^{-1}\right)=B_k(\partial_{x'})f=0\,.
\end{align*}
In the last identity we used the fact that $B_k$ has no constant term (cf. Remark \ref{rem:odd-d} and \eqref{eq:sesquiadj2})
and \eqref{20210706:eq3}.
This shows that $f$ does not depend on $x'$ and $t_k'$, thus we can set $x'=x$ and $t_k'=t_k$ in its definition \eqref{eq:f} to get
$$
f=f|_{x'=x,t_k'=t_k}=\res_z(S(z)S(-z)z^{-1})=1\,.
$$
This concludes the proof.
\end{proof}
In analogy with the proof of Lemma \ref{lem1-BC} setting $x'=x$ and $t_k'=t_k-\frac{2}{kw^k}$ in \eqref{bilBKP} we get,
using \eqref{eq:exp-odd}
$$
\Res_z\big(
(1+\frac{w}{z})\iota_w(w-z)^{-1}
S(z)
e^{-2w^{-1}\cdot\tilde\partial_{\bm t}}
S(-z)
\big)
=1
\,,
$$
which, by the second equation in \eqref{eq:positive}, gives
\begin{equation}\label{20170910:eq4b-odd-bis}
\Big(
(1+\frac{w}{z})
S(z)
e^{-2w^{-1}\cdot\tilde\partial_{\bm t}}
S(-z)
\Big)_-\Big|_{z=w}
=1
\,.
\end{equation}
Note that the LHS of \eqref{20170910:eq4b-odd-bis} can be rewritten as
\begin{equation}\label{20170910:eq4b-odd-bis2}
\begin{split}
&2
S(w)
e^{-2w^{-1}\cdot\tilde\partial_{\bm t}}
S(-w)
-\Big(
(1+\frac{w}{z})
S(z)
e^{-2w^{-1}\cdot\tilde\partial_{\bm t}}
S(-z)
\Big)_+\Big|_{z=w}
\\
&=
2
S(w)
e^{-2w^{-1}\cdot\tilde\partial_{\bm t}}
S(-w)
-1\,.
\end{split}
\end{equation}
Combining equations \eqref{20170910:eq4b-odd-bis} and \eqref{20170910:eq4b-odd-bis2} we get
\begin{equation}\label{20170910:eq4b-odd-bis3}
S(w)
e^{-2w^{-1}\cdot\tilde\partial_{\bm t}}
S(-w)
=
1\,.
\end{equation}
Next,
setting $x'=x$ and $t_k'=t_k-\frac2{kw_1^k}-\frac2{kw_2^k}$ in \eqref{bilBKP}, we have, using \eqref{20170910:eq7-odd},
\begin{align*}
& \res_z\left(S(z)e^{-2w_1^{-1}\cdot\tilde\partial_{\bm t}}e^{-2w_2^{-1}\cdot\tilde\partial_{\bm t}}S(-z)z^{-1}\right)
\\
& =1+
2\frac{w_1+w_2}{w_1-w_2}\res_z
\Big(
\left(w_1\iota_{w_1}(w_1-z)^{-1}-w_2\iota_{w_2}(w_2-z)^{-1}\right)\times \\
&\qquad\qquad\times
S(z)e^{-2w_1^{-1}\cdot\tilde\partial_{\bm t}}e^{-2w_2^{-1}\cdot\tilde\partial_{\bm t}}S(-z)z^{-1}
\Big)
\,.
\end{align*}
Since $\res_z\left(S(z)e^{-2w_1^{-1}\cdot\tilde\partial_{\bm t}}e^{-2w_2^{-1}\cdot\tilde\partial_{\bm t}}S(-z)z^{-1}\right)=1$, the above equation becomes
\begin{equation}\label{eq:x2}
\begin{split}
& \res_z
\Big(
\left(w_1\iota_{w_1}(w_1-z)^{-1}-w_2\iota_{w_2}(w_2-z)^{-1}\right)
\\
&\qquad \times
S(z)e^{-2w_1^{-1}\cdot\tilde\partial_{\bm t}}e^{-2w_2^{-1}\cdot\tilde\partial_{\bm t}}S(-z)z^{-1}
\Big)
=0\,.
\end{split}
\end{equation}
Using the second equation in \eqref{eq:positive} we can rewrite \eqref{eq:x2} as
\begin{equation}\label{20210628:eq6bkp}
\begin{split}
&
w_1\left(
S(z)e^{-2w_1^{-1}\cdot\tilde\partial_{\bm t}}e^{-2w_2^{-1}\cdot\tilde\partial_{\bm t}}S(-z)z^{-1}
\right)_-\Big|_{z=w_1}
\\
&
=
w_2\left(
S(z)e^{-2w_1^{-1}\cdot\tilde\partial_{\bm t}}e^{-2w_2^{-1}\cdot\tilde\partial_{\bm t}}S(-z)z^{-1}
\right)_-\Big|_{z=w_2}
\,.
\end{split}
\end{equation}
The $-$ sign means that we need to take negative powers in $z$ before substituting $z=w_1$ in the first parenthesis
and $z=w_2$ in the second.
We next observe that both terms in parenthesis belong to $\mc F[[z^{-1}]]z^{-1}$, hence we can remove the $-$ sign and get
$$
S(w_1)e^{-2w_1^{-1}\cdot\tilde\partial_{\bm t}}e^{-2w_2^{-1}\cdot\tilde\partial_{\bm t}}S(-w_1)
=
S(w_2)e^{-2w_1^{-1}\cdot\tilde\partial_{\bm t}}e^{-2w_2^{-1}\cdot\tilde\partial_{\bm t}}S(-w_2)\,,
$$
which, using \eqref{20170910:eq4b-odd-bis3}, can be rewritten as
\begin{equation}\label{cccp}
\frac{S(w_1)}{e^{-2w_2^{-1}\cdot\tilde\partial_{\bm t}}S(w_1)}
=
\frac{S(w_2)}{e^{-2w_1^{-1}\cdot\tilde\partial_{\bm t}}S(w_2)}\,.
\end{equation}
Taking logarithm of both sides of \eqref{cccp} and applying $N(w_1)N(w_2)$ equation \eqref{eq:proof-tau2-BCb} for $n=1$ follows
in the same way as in the proof of Lemma \ref{lem2}(b).

\end{document}